\documentclass[11pt]{article}
\usepackage[a4paper, total={6.5in, 9.5in}]{geometry}
\usepackage[utf8]{inputenc}
\usepackage[T1]{fontenc}
\usepackage{amsmath, amssymb, amsthm}
\usepackage{tikz}

\allowdisplaybreaks[1]

\theoremstyle{plain}
\newtheorem{theorem}{Theorem}

\newtheorem{lemma}{Lemma}
\newtheorem{corollary}{Corollary}
\theoremstyle{definition}
\newtheorem{definition}[theorem]{Definition}
\newtheorem{example}[theorem]{Example}
\theoremstyle{remark}

\newtheorem*{acks}{Acknowledgments}

\newcommand{\conv}{\otimes}
\newcommand{\deconv}{\oslash}

\newcommand{\E}{\ensuremath{\mathbf{E}}}
\newcommand{\p}{\ensuremath{\mathbf{P}}}
\newcommand{\R}{\ensuremath{\mathbb{R}}}
\newcommand{\N}{\ensuremath{\mathbb{N}}}
\newcommand{\bu}{\bullet}

\usepackage{subfig}

\author{
	Anne Bouillard\\
	Huawei Technologies France\\
	\texttt{anne.bouillard@huawei.com}
	\and
	Paul Nikolaus, Jens Schmitt\\
	Distributed Computer Systems (DISCO) Lab\\
	TU Kaiserslautern\\
	\texttt{\{nikolaus,jschmitt\}@cs.uni-kl.de}
}
\date{}

\title{Unleashing the Power of Paying Multiplexing Only Once in Stochastic Network Calculus}

\begin{document}
	\maketitle
	\begin{abstract}
		The stochastic network calculus (SNC) holds promise as a framework to calculate probabilistic performance bounds in networks of queues. 
		A great challenge to accurate bounds and efficient calculations are stochastic dependencies between flows due to resource sharing inside the network. 
		However, by carefully utilizing the basic SNC concepts in the network analysis the necessity of taking these dependencies into account can be minimized. 
		To that end, we fully unleash the power of the pay multiplexing only once principle (PMOO, known from the deterministic network calculus) in the SNC analysis. 
		We choose an analytic combinatorics presentation of the results in order to ease complex calculations. 
		In tree-reducible networks, a subclass of a general feedforward networks, we obtain a perfect analysis in terms of avoiding the need to take internal flow dependencies into account. 
		In a comprehensive numerical evaluation, we demonstrate how this unleashed PMOO analysis can reduce the known gap between simulations and SNC calculations significantly, and how it favourably compares to state-of-the art SNC calculations in terms of accuracy and computational effort. 
		Driven by these promising results, we also consider general feedforward networks, when some flow dependencies have to be taken into account. 
		To that end, the unleashed PMOO analysis is extended to the partially dependent case and a case study of a canonical example topology, known as the diamond network, is provided, again displaying favourable results over the state of the art.
	\end{abstract}

	\section{Introduction}
\label{sec:intro}
Stochastic network calculus (SNC) is a promising uniform framework to calculate probabilistic end-to-end performance bounds for individual flows in networks of queues. The most prominent goal is to control tail probabilities for the end-to-end (e2e) delay, i.e., probabilities for rare events shall be bounded, e.g., $ \p({\text{e2e delay} > 10 \text{ms}}) \leq 10^{-6} $.
Many modern systems are eager after such guarantees, as exemplified in visions like, e.g. Tactile Internet \cite{Fettweis14} or Industrial IoT \cite{BHCW18}.

SNC originates in the deterministic analysis by Rene Cruz \cite{Cruz91_1, Cruz91_2} and was subsequently supplemented by the use of min-plus algebra \cite{BCOQ92}.
In the following years, it was transferred to a stochastic setting \cite{Chang00, CBL06, Fidler06, JL08, CS12-1}.
Over the course of almost 30 years, two main branches of SNC have evolved: either by characterizing arrivals and service by envelope functions / tail bounds \cite{Cruz96, CBL06, JL08}, or, by moment-generating function  (MGF) bounds \cite{Chang00, Fidler06}.
While a larger class of processes can be modelled with tail bounds, \cite{RF11} comes to the conclusion that using MGFs leads to tighter bounds under the assumption of independence.

It should be mentioned that the \textit{uniform} approach of SNC, in particular to apply 
the union bound to evaluate sample-path events, comes at a price: already in the single-node case there is a known gap between simulations and SNC calculations (see Figure~\ref{fig:single-node-intro} on the next page for some typical numerical results). In fact, there is a tight analysis for the single-node case for some traffic classes based on martingale techniques \cite{CPS14, PC14} (avoiding the use of the union bound);  yet, an end-to-end martingale analysis remains an elusive goal. Therefore, in this paper, we keep following the uniform approach of SNC and concentrate our efforts on not widening the simulation-calculation gap further when analysing larger and more complex networks.

Analysing more general networks of queues, in particular feedforward networks, usually requires the consideration of stochastically dependent flows. 
Even if all external arrival and service processes are independent, the sharing of resources by individual flows at queues generally makes them stochastically dependent at subsequent queues. 
How much this kind of dependencies has to be taken into account is affected by the network analysis method because different methods require different levels of knowledge about the internal characterization of flows.
Further on, we call these dependencies \textit{method-pertinent}. 
To deal with (method-pertinent) stochastic dependencies in SNC, typically, the MGF of, e.g. the sum of dependent flows is upper bounded by H\"older's inequality (HI).
While HI's generality is convenient and consistent with SNC's aspiration as a uniform framework, it often incurs the drawback of degrading bound accuracy considerably, further widening the simulation-calculation gap.
In addition, it increases the computational effort by introducing an additional parameter to optimize for each application of HI, such that in larger scenarios runtimes quickly become prohibitive (see also Subsection~\ref{subsec:extend-interleaved-tandem}). 
Consequently, SNC analysis methods with less method-pertinent dependencies are strongly favourable as they require less invocations of HI.
In fact, previous work in relatively simple network scenarios has indicated that techniques which completely avoid HI achieve significantly better delay bounds \cite{ZBH16, NS17-1}.


\begin{figure*}[tbp]		
	\centering
		\subfloat[Original topology \label{fig:intro-toy}]{
		\resizebox{0.38\textwidth}{!}{%
		\begin{tikzpicture}[server/.style={shape=rectangle,draw,minimum height=.8cm,inner xsep=3ex}]
			\node[server,name=S1] at (0,0) {$1$};
			\node[server,name=S2] at (2,0) {$2$};
			\node[server,name=S3] at (4,0) {$3$};
			\draw[->, thick, red] (-1,0) node[left] {$1$} -- (5,0);
			\draw[->, thick, blue] (-1,0.3) node[left] {$2$} to[out=-10,in=-170] (2.7,.3);
			\draw[->, thick, green!70!black]  (1,-0.3) node[left] {$3$} to[out=10,in=170] (5,-.3);
		\end{tikzpicture}
	}%
	}
	\hspace{5mm}
	\subfloat[\centering Reduction to end-to-end server 
	\hspace{\textwidth} after network analysis \label{fig:e2e-server}]{
	\includegraphics[width=0.42\textwidth]{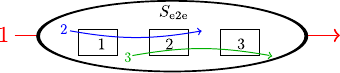}
	}
	\vspace{-2mm}
	\caption{Interleaved tandem network \label{fig:interleaved-tandem-nw-analysis}}
	\vspace{-2mm}
\end{figure*}
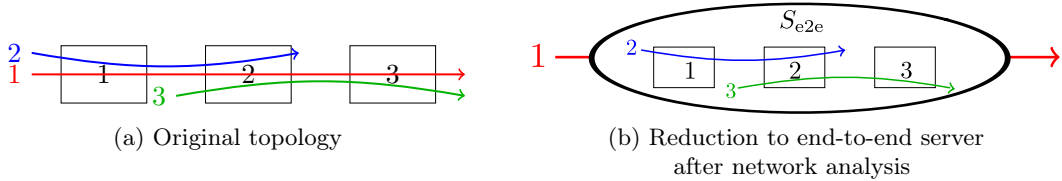

\begin{figure*}[b]
	\vspace{-5mm}
	\centering
	\subfloat[Single-node case \label{fig:single-node-intro}]{
		\includegraphics[width=0.3\textwidth]{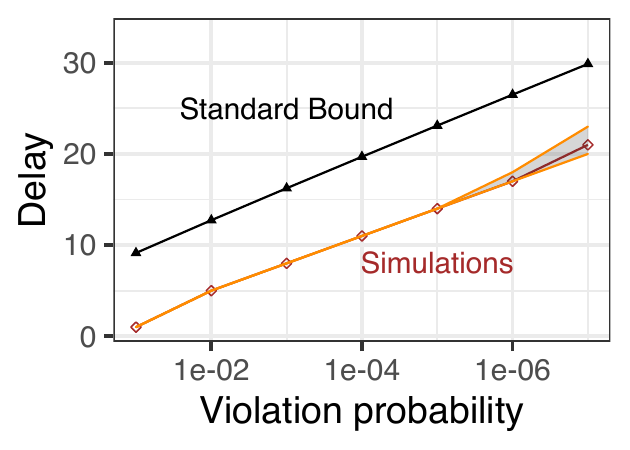}
	}
	\hspace{5mm}
	\subfloat[Interleaved tandem \label{fig:interleaved-tandem-intro}]{
		\includegraphics[width=0.3\textwidth]{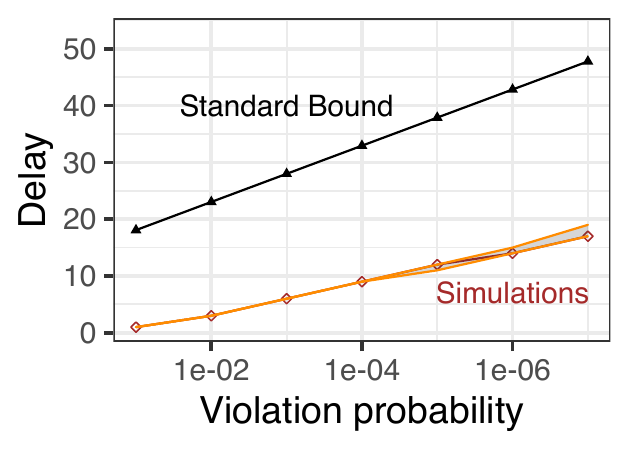}
	}
	\vspace{-1mm}
	\caption{SNC delay bounds and simulation results (for traffic with exponentially distributed increments and constant rate servers). For more details, see Section~\ref{sec:numerical}. \label{fig:delay-bounds-intro}}
\end{figure*}

Different network analysis methods have been investigated intensively in the deterministic setting; some better known ones are, e.g. Separate Flow Analysis (SFA) or Pay multiplexing only once (PMOO) \cite{SZ06-1}.
These typically try to reduce the analysis (e.g. for the delay) of a particular flow in the network to a simple analysis of this flow traversing a single server; the characteristics of that so-called end-to-end server depend on the cross-flows and the servers of the original network, see also Figure~\ref{fig:interleaved-tandem-nw-analysis}.
The main goal for the network analysis in the deterministic case was to properly account for the bursts of the flows to ensure that they are only "paid once". In particular, the PMOO
 principle \cite{SZ06-1, SZM08-1, BGLT08, Fidler03} tries to ensure this for all flows in the network by, whenever possible, concatenating servers first before calculating residual service by "subtracting" cross-flows. As discussed above, in SNC the main concern is dependencies, yet it turns out that, to some extent, this is related to the proper accounting of bursts. 
In fact, in \cite{NS17-1}, it has been observed that the PMOO analysis known from deterministic network calculus
\cite{SZ06-1, SZM08-1, BGLT08, Fidler03} leads to less method-pertinent dependencies compared to SFA and is thus also promising in an SNC analysis.

However, the application of PMOO in the SNC has, so far, been limited to so-called nested interference structures -- this is very restrictive.
For instance, for a network with interleaved interference as in Figure~\ref{fig:intro-toy}, state-of-the-art analysis still requires at least one application of HI, even if we assume all external arrival processes to be independent. 
This is also illustrated in Figure~\ref{fig:interleaved-tandem-intro} (anticipating some of the results from Section~\ref{sec:numerical}): 
method-pertinent dependencies force the state-of-the-art SNC bounds to deteriorate significantly; in particular, it is observable that, in comparison to the single-node case, even the scaling of the delay bounds is not captured correctly any more and the simulation-calculation gap widens considerably even for such a small network.

The overall goal of our paper is therefore to unleash the power of the PMOO principle in the SNC framework in order to not widen the simulation-calculation gap further even in more complex and larger networks of queues.
To that end, we make the following contributions:
\begin{itemize}
	\item We present a PMOO-based SNC end-to-end analysis for a subclass of feedforward networks, so-called tree-reducible networks; the main result is given in Theorem~\ref{th:e2e-sc}.
	It achieves \textit{zero} method-pertinent stochastic dependencies 
	when external arrivals and service processes are independent.
	I.e., if all input flows are assumed to be independent, we can derive bounds without using H\"older's inequality. 
	Also, Theorem~\ref{th:e2e-sc} allows us to calculate the residual service in one big step avoiding the  sequencing penalty in previous network analysis methods.
	\item We apply analytic combinatorics \cite{FS09} to recover bounds from state-of-the-art analysis methods in simple networks and enable a generalization to more complex settings.
\item We conduct an extensive numerical evaluation with respect to the accuracy of the new bounds for several traffic classes and different network topologies.
	\item We discuss first results to extend our new method from tree-reducible to general feedforward networks, still striving for the goal to minimize method-pertinent dependencies.
\end{itemize}

\vspace{-1mm}

	\section{MGF-based Network Calculus and Analytic Combinatorics}
\label{sec:framework}
In this section, we briefly introduce the two main tools we employ in our paper, namely the MGF-based stochastic network calculus in Subsection~\ref{sec:mgf-snc} and analytic combinatorics in Subsection~\ref{sec:comb}. Furthermore, we show in Subsection~\ref{sec:snc-comb} how analytic combinatorics can be used to express complex calculations in MGF-based SNC in a concise manner. 

\vspace{-2mm}

\subsection{Stochastic network calculus framework}
\label{sec:mgf-snc}
We present the main concepts of the moment-generating function (MGF)-based network calculus for the single-server case (the network case is postponed to Section~\ref{sec:pmoo} in the context of our new contribution).
More details can be found in~\cite{Chang00, Fidler06}.
We assume time is discrete and space is continuous. 
We deal with bivariate functions and always assume that their definition domain is $\{(s, t)\in\N^2~|~s\leq t\}$ and that they are non-negative: $f(s, t) \geq 0$ for all $s\leq t$. 

\paragraph{Bivariate arrival processes}
A bivariate (arrival) process $A$ of a flow at some point in the network represents the amount of data of the flow  traversing that point of the network during any interval of time: let $a_i$ be the amount of data during the $i$-th time slot; we define for all $s\leq t$, $A(s, t) = \sum_{i=s+1}^t a_i$, with the convention $A(t, t) = 0$. Note that bivariate arrival processes are additive: for all $s\leq u\leq t$, $A(s, u) + A(u, t) = A(s, t)$.

\paragraph{(min, plus) operations} To describe the transformation between bivariate processes when traversing a server, we rely on the following operations in the (min,plus)-algebra. Let $f$ and $g$ be two bivariate functions. 

\begin{itemize}
	\item \textit{(min, plus)-convolution:} $f\conv g(s, t) = \min_{s\leq u\leq t} f(s, u) + g(u, t)$;
	\item  \textit{(min, plus)-deconvolution:} $f\deconv g(s, t) = \max_{0\leq u\leq s} f(u, t) - g(u, s)$.
\end{itemize} 

\paragraph{$S$-servers} Let $S$ be a non-negative bivariate function. A server is a dynamic $ S $-server if the relation between its bivariate arrival and departure processes $A$ and $D$ satisfies for all $t\geq 0$, $A(0, t)\geq D(0, t) \geq A\conv S(0, t)$. 

This notion of server can be too weak in some situations when performing a network analysis, and we introduce work-conserving $S$-servers: assume that $S(t,t) = 0$ for all $t$ and let $start(t) = \sup \left\{s\leq t~|~A(0, s) = D(0, s)\right\}$ be the last instant before $t$ when the server is empty. We also call it the \textit{start of the backlogged period} of $t$. 
Note that from this definition, if the server is idle at time $t$, then $start(t) = t$.
We say that the $S$-server is \textit{work-conserving} if for all $t$ and $s\in\{start(t),\ldots,  t\}$,  $D(0, t) \geq D(0, s) + S(s, t)$. In other words, the service offered between times $s$ and $t$ is at least $S(s, t)$, provided that the server is always backlogged between times $s$ and $t$. 
We can always assume that $S$ is additive, whereas this does not have to hold for dynamic $S$-servers. 

The notions of dynamic $S$-server and of work-conserving $S$-server, respectively, correspond to the notions of service curve and strict service curve in deterministic network calculus \cite{BBLC18}. Note that a work-conserving $S$-server is also a dynamic $S$-server. 
	
\paragraph{Departure process characterization and performance bounds}
Consider a dynamic $S$-server and $A$ and $D$ its respective bivariate arrival and departure processes. 

The backlog at time $t$ is $q(t) = A(0, t) - D(0, t)$ and the virtual delay at time $t$ is $d(t) = \inf\{ T \in\N \mid A(0,t) \leq D(0, t + T)\}$.  

\begin{theorem}[Sample-Path Bounds~\cite{Chang00, Fidler06}]
	\label{thm:sample-path-bounds}
	Consider a dynamic $S$-server traversed by a flow with bivariate arrival process $A$. Then,
	\begin{itemize}
		\item $D(s, t) \leq A\deconv S(s, t)$;
		\item $q(t) \leq A\deconv S(t, t)$;
		\item $d(t) \geq T \Rightarrow \exists s \leq t,~A(s, t) > S(s, t+T - 1).$
	\end{itemize}
\end{theorem}

\paragraph{$(\sigma, \rho)$-constraints}
MGF-based SNC relies on bounds on the MGF and the Laplace transform for the arrival and service process, respectively. 
Specifically, we assume that $A$ and $S$ are bivariate stochastic processes, and that $\E[e^{\theta A(s,t)}]$ and $\E[e^{-\theta S(s,t)}]$ both exist for some $ \theta > 0. $ 
We say that  the bivariate arrival process $A$ is $(\sigma_A, \rho_A)$-constrained if there exist $\theta>0$ and $\sigma_A(\theta), \rho_A(\theta)\in\R_+ $ 
 such that $\E[e^{\theta A(s, t)}] \leq e^{\theta(\sigma_A(\theta) + \rho_A(\theta)(t-s))}$. 
Similarly, we say that  the bivariate service process $S$ is $(\sigma_S, \rho_S)$-constrained if there exist $\theta>0$ and $\sigma_S(\theta), \rho_S(\theta)\in\R_+$ 
such that $\E[e^{-\theta S(s, t)}] \leq e^{\theta(\sigma_S(\theta) - \rho_S(\theta)(t-s))}$.

Such linear bounds in the transform space provide the advantage that they are closed with respect to network calculus operations, e.g. multiplexing of flows.
However, more general bounds are not fundamentally precluded and in some of the following results we also provide this generality when it does not harm the clarity of the presentation.


	\vspace{-2mm}

	\subsection{Analytic combinatorics framework}
	\label{sec:comb}
	Let $f = (f_n)_{n\in\N} \in\R_+^{\N}$ be a sequence of non-negative numbers. The (ordinary) generating function (related to the $z$-transform) associated to $f$ is the function $F(z) = \sum_{n = 0}^\infty f_n z^n$.  An important example in the following is the geometric sequence $f_n = a^n$, for which
	\begin{equation}
		\label{eq:geom}
	F(z) = \sum_{n = 0}^\infty a^nz^n = \frac{1}{1- az}.	
	\end{equation}
	Given a generating function $F(z)$, we denote the $n$-th term of its associated sequence by $[z^n]F(z)$. In other words, with the previous notation, $[z^n]F(z) = f_n$.  
	Analytic combinatorics~\cite{FS09} is a branch of combinatorics whose aim is the study of such sequences, and, in particular, the relation between the singularities of the generating function and the asymptotic behaviour of $f$. 
	We denote by $r_F$ the dominant singularity of $F$. In the following, we only use rational fractions, and the dominant singularity corresponds to the smallest root of the denominator. For example, the dominant singularity of $F(z)$ in Equation~\eqref{eq:geom} is $r_F =  a^{-1}$.
	 
	Next, we apply the framework of analytic combinatorics to simplify some SNC calculations, in particular, when computing error bounds for backlogs and delays. 
	Let us first present some operations on sequences and their associated generating functions. 

	\paragraph{Cauchy product: } Let $f = (f_n)_{n\in\N}$ and $g = (g_n)_{n\in\N}$ be two sequences. The Cauchy product of these sequences is $h_{\mathrm{Cau}} = (h_{\mathrm{Cau},n})_{n\in\N}$ with $h_{\mathrm{Cau},n} = \sum_{p=0}^n f_pg_{n-p}$. 
	The generating function of $h_\mathrm{Cau}$ is $$H_\mathrm{Cau}(z) = \sum_{n =  0}^\infty \sum_{p+q = n} f_pg_q z^n = (\sum_{p = 0}^\infty f_p z^p) (\sum_{q = 0}^\infty g_q z^q) = F(z)G(z),$$
	and is also called the Cauchy product of $F$ and $G$. 
	The dominant singularity of $H_\mathrm{Cau}$ is $r_{H_\mathrm{Cau}} = \min(r_F, r_G)$.
	In the SNC, the (min,plus)-convolution corresponds to a Cauchy product.  
	
	\paragraph{Hadamard product: }  Let $f = (f_n)_{n\in\N}$ and $g = (g_n)_{n\in\N}$ be two sequences. The Hadamard product of these sequences is $h_\mathrm{Had} = (h_{\mathrm{Had},n})_{n\in\N}$ with $h_{\mathrm{Had},n} = f_n g_n$.
	 There is no simple expression for the generating function of the Hadamard product in the general case. 
	 Yet, in the case where $f = (a^n)_{n\in\N}$ is a geometric sequence, the Hadamard product of $F$ and $G$ is
	\begin{equation}
		\label{eq:hadamard}
		H_\mathrm{Had}(z) = \sum_{n = 0}^\infty a^n g_n z^n = G(az).
	\end{equation}
The dominant singularity of $H_\mathrm{Had}$ is $r_{H_\mathrm{Had}} = r_F \cdot r_G$.
In the SNC, the Hadamard product is used when computing the (min,plus)-deconvolution.
	
	\paragraph{Sum of last  terms:} Let $f = (f_n)_{n\in\N}$ be a summable sequence. Let $g_n = \sum_{m \geq n} f_m$. The generating function of $g = (g_n)_{n\in\N}$, proved in Appendix~\ref{sec:elie}, is
	\begin{equation}
		\label{eq:elie}
		G(z) = \frac{F(1) - zF(z)}{1-z}.
	\end{equation}
	Note that 1 is not a singularity of $G(z)$ if it is not a singularity of $F$. Consequently, $F$ and $G$ have the same dominant singularity.
	This operation is used when computing delay bounds.

	
	\vspace{-2mm}
	
	\subsection{A combinatorial view of MGF-based network calculus}
	\label{sec:snc-comb}
	The purpose of this subsection is to present the $(\sigma, \rho)$-constraints on arrival and service processes as generating functions. This alleviates the calculation of residual service constraints in complex scenarios, thus enabling accurate performance bounds. 
	The use of analytic combinatorics in not new in the field of queueing theory. For example, Flajolet and Guillemin in \cite{FG00} analyze Markovian queues, Guillemin and Pinchin in~\cite{GP04}  study the processor  sharing policy, and more recently,~\cite{BCPM18} presents an attempt to use analytical combinatorics in Network Calculus. The approach of all these references is to model the distributions of the queues by a generating functions and use the machinery to compute the asymptotic behavior of the queue. Here, we use generating functions only for the purpose of simplifying the computations of the SNC framework. To our knowledge, this is the first time this tool is used for this purpose.
	
	\begin{definition}
		Let $A$ be a bivariate arrival process. An \textit{arrival bounding generating function} of $A$ at $\theta$ is a function $F_A(\theta, z)$ such that for all $s, n\geq 0$, 
		$\E[e^{\theta A(s, s+n)}] \leq [z^n]F_A(\theta, z) $.
	\end{definition}

	For example, if $A$ is $(\sigma_A,\rho_A)$-constrained and $\sigma_A(\theta), \rho_A(\theta)\in\R_+$, 
	\[
		F_A(\theta, z) = \sum_{n = 0}^\infty e^{\theta(\sigma_A(\theta) + \rho_A(\theta)n)} z^n = \frac{e^{\theta\sigma_A(\theta)}}{1-e^{\theta\rho_A(\theta)}z}
	\]
	is an arrival bounding generating function for $A$ at $\theta$.
	
	Similarly, if $S$ is a service process,  a \textit{service bounding generating function} of $S$ at $\theta$ is a function $F_S(\theta, z)$, such that for all $s, n\geq 0$, 
	$\E[e^{-\theta S(s, s+n)}] \leq [z^n]F_S(\theta, z)
	$. 
	
	In the following, it shall be clear from the context whether a bounding generating function is an arrival or a service one. 
	Therefore, we omit this precision for the sake of readability.
	
	Using bounding generating functions of $A$ and $S$ at $\theta$,
	 we can derive bounding generating functions for the departure process and the violation probability of the probabilistic delay bound. The backlog bound can be directly deduced from the departure process. 
	The following results provide bounding generating functions for the departure process and the performance bounds.


	
	\begin{lemma}[Departure process bounding generating function]
		\label{lem:output}
		Consider a dynamic $S$-server traversed by a flow with bivariate process $A$ that is $(\sigma_A, \rho_A)$-constrained. Assume that $S$ has the service bounding generating function $F_S$  and that  the arrival and service processes $A$ and $S$ are independent.
		Then, for all $\theta$ such that $e^{-\theta \rho_A(\theta)} r_S(\theta) < 1$, an arrival bounding generating function of the departure process $D$ is
		$$F_D(\theta, z) = \frac{e^{\theta\sigma_A(\theta)}F_S(\theta, e^{\theta\rho_A(\theta)})}{1-e^{\theta\rho_A(\theta)}z}.$$
	\end{lemma}

	\begin{proof}
		As $F_A(\theta, z)$ is geometric, we have 
		$[z^{n+m}]F_A(\theta, z) = e^{\theta\rho_A(\theta)m}[z^n]F_A(\theta, z)$.
		Since $D(s, t) \leq A\deconv S(s, t) = \sup_{u\leq s} A(u, t) - S(u, s)$, then for $\theta >0$,
		\begin{align*}
			\E[e^{\theta A\deconv S(s, t)}] & \leq \sum_{u\leq s}  \E[e^{\theta (A(u, t) - S(u, s))}]\\
			&= \sum_{u\leq s}  \E[e^{\theta A(u, t)}]\cdot \E[e^{ - \theta S(u, s)}]\\
			&\leq \sum_{u\leq s}  [z^{t-u}]F_A(\theta, z)\cdot [z^{s-u}]F_S(\theta, z)  \\
			&=  e^{\theta\rho_A(\theta) (t-s)} \sum_{u\leq s}  [z^{s-u}]F_A(\theta, z)\cdot [z^{s-u}]F_S(\theta, z)\\ 
			&\leq  e^{\theta\rho_A(\theta) (t-s)} \sum_{n\geq 0}  [z^n]F_A(\theta, z)\cdot [z^n]F_S(\theta, z) \\ 
			&= e^{\theta\rho_A(\theta) (t-s)} e^{\theta\sigma_A(\theta)} F_S(\theta, e^{\theta\rho_A(\theta)})\\
			& = [z^{t-s}] F_A(\theta, z) F_S(\theta, e^{\theta\rho_A(\theta)}).
		\end{align*} 
		To obtain the second to last equality, 
		we applied the Hadamard product to a geometric series as in Equation~\eqref{eq:hadamard} with $z=1$.
		We now recognize the $ (t-s) $-th term of a geometric series with ratio $e^{\theta\rho_A(\theta)}$, so $$F_D(\theta, z) = \frac{e^{\theta\sigma_A(\theta)}F_S(\theta, e^{\theta\rho_A(\theta)})}{1-e^{\theta\rho_A(\theta)}z}$$
		is a bounding generating function for the departure process. 
	\end{proof}
	\begin{corollary}[Probabilistic backlog bound]
		\label{lem:backlog}
		Let us denote $q(t)$ as the backlog at time $t$. 
		Then, under the same assumptions as in Lemma~\ref{lem:output},
		$$\p(q(t) \geq b)  \leq e^{-\theta b} e^{\theta\sigma_A(\theta)} F_S(\theta, e^{\theta\rho_A(\theta)}).$$
	\end{corollary}

\begin{proof}
	As from Theorem~\ref{thm:sample-path-bounds} we have $q(t) \leq A\deconv S(t, t)$, we can apply the Chernoff bound using the first term of the generating function $F_D(\theta, z)$ of Lemma~\ref{lem:output}:  
	\begin{align*}
		\p(q(t) \geq b) & \leq \E[e^{\theta A\deconv S(t, t)}]e^{-\theta b} \\
		&\leq e^{\theta\sigma_A(\theta)} F_S(\theta, e^{\theta\rho_A(\theta)})e^{-\theta b}.
	\end{align*}
\end{proof}
In the particular case of $F_S(\theta, z) =\frac{e^{\theta\sigma_S(\theta)}}{1-e^{-\theta\rho_S(\theta)}z}$, for all $\theta\geq 0$ such that $\rho_A(\theta) < \rho_S(\theta)$,
$$\p(q(t) \geq b)  \leq \frac{e^{\theta(\sigma_A(\theta) + \sigma_S(\theta))}}{1-e^{\theta(\rho_A(\theta) - \rho_S(\theta))}} e^{-\theta b}.$$

We have made the assumption of $(\sigma_A, \rho_A)$-constraints for the arrival processes, and our proofs rely on this assumption. Yet, we remark that the calculations can be done in a more general setting, at the price of more complex formulas. For example, if the arrival bounding generating function of the arrival process $A$ is a sum of geometric series, as for example for Markov-modulated arrivals, the calculation directly follows, and the arrival bounding generating function of $D$ would be a sum of geometric series as well. Another case where computations can be adapted is when the service is bounded by a geometric series: symmetric computations can be done.

\begin{lemma}[Delay bound generating function]
	\label{lem:delay}
	Let us denote $d(t)$ as the virtual delay at time $t$. 
	Under the same assumptions as in Lemma~\ref{lem:output}, 
	a bounding generating function for the violation probability of the probabilistic delay bound is
	\begin{equation}
		F_d(\theta, z) = e^{\theta \sigma_A(\theta)} \frac{e^{\theta\rho_A(\theta)}F_S(\theta, e^{\theta\rho_A(\theta)}) -  zF_S(\theta, z)}{1-e^{-\theta\rho_A(\theta)}z}.
	\end{equation}
\end{lemma}
A proof can be found in Appendix~\ref{sec:delay}.
In the particular case of $F_S(\theta, z) =\frac{e^{\theta\sigma_S(\theta)}}{1-e^{-\theta\rho_S(\theta)}z}$, for all $\theta \geq 0$ such that $\rho_A(\theta) < \rho_S(\theta)$,
$F_d(\theta, z) = \frac{e^{\theta(\sigma_A(\theta)+\sigma_S(\theta) + \rho_A(\theta))}}{1-e^{-\theta(\rho_S(\theta) - \rho_A(\theta))}} \cdot \frac{1}{1-e^{-\theta \rho_S(\theta)}z}$, 
so $$\p(d(t)\geq T) \leq \frac{e^{\theta(\sigma_A(\theta)+\sigma_S(\theta) + \rho_A(\theta))}}{1-e^{-\theta(\rho_S(\theta) - \rho_A(\theta))}} e^{-\theta\rho_S(\theta)T}.$$ 

\bigskip

So far, we have focused on the case of independent stochastic processes, let us now release this assumption.
Since our analysis is based on moment-generating functions, the common approach in the literature is to upper bound the product of dependent processes by Hölder's inequality:
let $X_1, \ldots, X_n \geq 0$ be integrable random variables;
then, for all positive numbers $q_1, \ldots, q_n$ such that $\sum_{i=1}^n \frac{1}{q_i} = 1$, 
\begin{equation}
	\label{eq:holder}
	\E[X_1 \cdots X_n] \leq \prod_{i=1}^n \E[X_i^{q_i}]^{\frac{1}{q_i}}.
\end{equation}
Note that $\E[e^{p\theta A(s, t)}]^{1/p} \leq  e^{p\theta(\sigma_A(p\theta) + \rho_A(p\theta)(t-s))(1/p)} = e^{\theta(\sigma_A(p\theta) + \rho_A(p\theta)(t-s))}$ and, similarly,\\
  $\E[e^{-q\theta S(s, t)}]^{1/q} \leq e^{\theta(\sigma_S(q\theta) - \rho_S(q\theta)(t-s))}$.

The following lemma contains the translation of an existing output bound from~\cite{BS12-1} to the combinatorial framework as well as backlog and delay bounds for the dependent case.
\begin{lemma}
	Consider a dynamic $S$-server traversed by a flow with bivariate process $A$ that is $(\sigma_A, \rho_A)$-constrained, and assume that the process $S$ is also $(\sigma_S, \rho_S)$-constrained. Then  
	for all $p, q$ such that $\frac{1}{p} + \frac{1}{q}= 1$ and $ \rho_A(p\theta) < \rho_S(q\theta) $,
	\begin{itemize}
		\item $F_D(\theta, z) = \frac{e^{\theta(\sigma_A(p\theta)+\sigma_S(q\theta))}}{1-e^{-\theta(\rho_S(q\theta) - \rho_A(p\theta))}} \cdot \frac{1}{1-e^{\theta \rho_A(p\theta)}z}$;
		\item $\p(q(t) \geq b)  \leq \frac{e^{\theta(\sigma_A(p\theta)+\sigma_S(q\theta))}}{1-e^{-\theta(\rho_S(q\theta) - \rho_A(p\theta))}} e^{-\theta b}$;
		\item  $\p(d(t)\geq T) \leq \frac{e^{\theta(\sigma_A(p\theta)+\sigma_S(q\theta) + \rho_A(p\theta))}}{1-e^{-\theta(\rho_S(q\theta) - \rho_A(p\theta))}} e^{-\theta\rho_S(q\theta)T}$.
	\end{itemize}
\end{lemma}

	\section{Tree Network Analysis under Complete Independence}
\label{sec:pmoo}

In this section, we derive a result enabling us to unleash the power of the PMOO principle for the SNC when all arrival and service processes are originally independent. In particular, this result is directly applicable to tree networks. However, given an overall network with servers and flows, it is only from the perspective of the flow of interest (foi) that the network has to constitute a tree. In other words, in a preliminary step, we (attempt to) reduce the network to a tree, with all associated arrival and service processes still being independent, by performing the following steps:

\begin{enumerate}
\item Terminate flows \textit{after} their last direct or indirect interaction with the foi
and reduce the network by deleting all flows with which the foi has neither a direct nor indirect dependency (also delete all servers that do not carry any flows any more).
\item Check for all remaining flows whether they directly or indirectly rejoin (a) the foi, or,  (b) each other.
\item If there is any rejoining flow in step (2) then the network is not tree-reducible, otherwise, it is and the remaining flows and servers form an in-tree whose root is the last server traversed by the foi.
\end{enumerate}

While not being perfectly formal, it should be clear from this procedure that tree-reducibility imposes clearly a restriction to general feedforward networks, simply speaking by prohibiting rejoining flows. Nevertheless, it covers quite a number of cases and generalizes significantly on previous work in SNC network analysis. In Section~\ref{sec:dependent}, we come back to the issue of general feedforward networks \textit{with} rejoining flows and how we can leverage from the PMOO result for tree networks in that case. 

Two elementary operations to analyse networks in network calculus are computing residual service offered to the flow of interest for a server crossed by multiple flows and the concatenation of servers when a flow traverses multiple servers. Let us recall the main results from the literature concerning MGF-based network calculus. 

\begin{theorem}[Residual server~\cite{Fidler06}]
	\label{thm:residual-ser}
	Consider a work-conserving $S$-server traversed by two flows with respective bivariate arrival processes $A_1$ and $A_2$. The residual server of flow 1 can be characterized by a dynamic $S_1$-server with 
	$$S_1(s, t) = [S(s, t) - A_2(s, t)]_+,~\forall s\leq t,$$
	where $[x]_+ = \max(0, x)$.
\end{theorem}
Here, we make no assumption about the service policy regarding the sharing between flows. 
This is also known as arbitrary or blind multiplexing.

\begin{theorem}[End-to-end server~\cite{Fidler06}]
	Consider a flow traversing a tandem of dynamic $S_i$-servers, $i \in \{1, \dots, n\}.$
The overall service offered is a dynamic $S_{\mathrm{e2e}}$-server with
$S_{\mathrm{e2e}} = S_1\conv \cdots \conv S_n$. 
\end{theorem}

In this section, we generalize and combine these two central theorems for tree networks. 
We first bound the service process for the end-to-end dynamic server offered to a flow traversing a network in Subsection~\ref{sec:biv-pmoo}. Next, in Subsection~\ref{sec:series-pmoo} we derive the corresponding bounding generating function and apply the framework from Subsection~\ref{sec:snc-comb} to derive performance bounds.

%

\vspace{-1mm}

\subsection{A Pay-Multiplexing-Only-Once formula for bivariate processes}	
\label{sec:biv-pmoo}
In this section, we consider a tree network. 
More precisely, the network can be described as follows. 

\begin{enumerate}
	\item The network is composed of $n$ servers, numbered from 1 to $n$ according to the topological order:  if server $h$ is the successor of server $j$, then $h>j$. The topology is a tree directed to server $n$: each server $j\neq n$ has a unique successor that we denote by  $j^\bu$. Each server $j$ is work-conserving and offers the service $S_j(s, t)$ during the time slots $s+1, \ldots, t$. By convention, we write $n^{\bu} = n+1$. 
	\item There are $m$ flows in the network, numbered from 1 to $m$. 
	We denote $\pi_i = \langle \pi_i(1), \ldots, \pi_i(\ell_i)\rangle$ the path, i.e. the sequence of servers, of length $\ell_i\in\N\setminus\{0\}$ followed by flow $i$ 
	and $A_i(s, t)$ is the amount traffic of flow $i$ arriving in the network during the time slots $s+1, \ldots, t$. We write $i\in j$ if flow $i$ crosses server $j$ (or equivalently $j\in\pi_i$).
	\item We denote $A^{(j)}_i(s, t)$ the bivariate process of flow $i$ at the input of server $j$, and consequently, $A^{(j^{\bu})}_i(s, t)$ is the bivariate process at the output of server $j$.  Note that we have $A_i = A^{(\pi_i(1))}_i$.
\end{enumerate}

\begin{example}
	The networks of Figures~\ref{fig:toy} and~\ref{fig:tree} are two examples of tree networks. They both have three flows and three servers. In the network of Figure~\ref{fig:toy}, we have $1^{\bu} = 2$, $2^{\bu} = 3$ and $3^{\bu} = 4$, whereas in the network of Figure~\ref{fig:tree}, we have $1^{\bu} = 2^{\bu} = 3$ and $3^{\bu} = 4$. 
	Graphically, the relation $j^{\bullet} = h$ means that $(j, h)$ is an edge of the network.
\end{example}

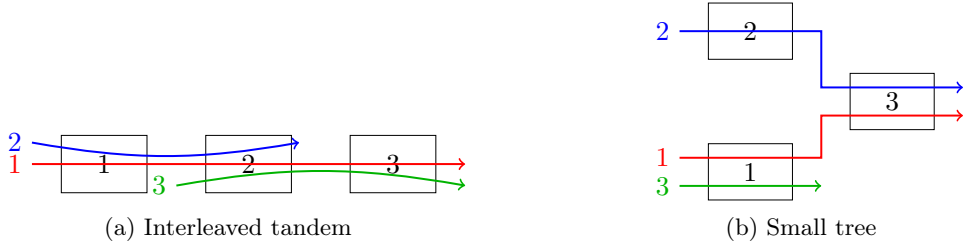
\begin{figure*}[t]		
	\centering
	\subfloat[Interleaved tandem \label{fig:toy}]{
		\resizebox{0.38\textwidth}{!}{%
		\begin{tikzpicture}[server/.style={shape=rectangle,draw,minimum height=.8cm,inner xsep=3ex}]
			\node[server,name=S1] at (0,0) {$1$};
			\node[server,name=S2] at (2,0) {$2$};
			\node[server,name=S3] at (4,0) {$3$};
			\draw[->, thick, red] (-1,0) node[left] {$1$} -- (5,0);
			\draw[->, thick, blue] (-1,0.3) node[left] {$2$} to[out=-10,in=-170] (2.7,.3);
			\draw[->, thick, green!70!black]  (1,-0.3) node[left] {$3$} to[out=10,in=170] (5,-.3);
		\end{tikzpicture}
	}%
	}
	\hspace{20mm}
	\subfloat[Small tree \label{fig:tree}]{
		\resizebox{0.26\textwidth}{!}{%
		\begin{tikzpicture}[server/.style={shape=rectangle,draw,minimum height=.8cm,inner xsep=3ex}]
			\node[server,name=S1] at (0,0) {$3$};
			\node[server,name=S2] at (-2,-1) {$1$};
			\node[server,name=S3] at (-2,1) {$2$};
			\draw[->, thick, red] (-3,-0.8) node[left] {$1$} -- (-1,-0.8) -- (-1, -0.2) -- (1, -0.2);
			\draw[->, thick, blue] (-3,1) node[left] {$2$} -- (-1,1) -- (-1, 0.2) -- (1, 0.2);
			\draw[->, thick, green!70!black]  (-3,-1.2) node[left] {$3$}  -- (-1,-1.2);
		\end{tikzpicture}
	}%
	}
	\vspace{-2mm}
	\caption{Two example networks.}
	\vspace{-3mm}
\end{figure*}

	

We can now state the main result of this section. Theorem~\ref{th:e2e-sc} also generalizes that of~\cite{BGLT08, SZ06-1, SZM08-1} in deterministic network calculus from tandems to trees. 
More precisely, all the results in those references assume that the flow for which we compute the residual service curve traverses the entire tandem network. Here we relax this assumption. However, the principle of the proof remains the same, and our result is directly applicable to deterministic network calculus. For the sake of self-containedness, we give a proof in Appendix~\ref{app:simple-pmoo}.

\begin{theorem}
	\label{th:e2e-sc}
	With the notations above, assume that the last server traversed by flow 1 is $n$. The residual service offered to flow 1 is a dynamic $S_{\mathrm{e2e}}$-server with $\forall 0\leq   t_{\pi_1(1)} \leq t_{n+1}$,  
%
		$$S_{\mathrm{e2e}}(t_{\pi_1(1)}, t_{n+1}) =  \left[\inf_{\forall j,~t_j \leq t_{j^\bu}} \sum_{j=1}^n [S_j(t_j, t_{j^\bu})] - \sum_{i = 2}^m A_i(t_{\pi_i(1)}, t_{\pi_i(\ell_i)^\bu})\right]_+.$$
\end{theorem}
Note that there is a parameter $ t_j $ for each server $ j $ following the topological order of the servers.

\begin{example}
	The end-to-end server for flow 1 for the network in Figure~\ref{fig:toy} is 
	$$S_{\mathrm{e2e}}(t_1, t_4) = \left[\inf_{t_1 \leq t_2\leq t_3\leq t_4} (S_1(t_1, t_2) + S_2(t_2, t_3) + S_3(t_3, t_4) - A_2(t_1, t_3) - A_3(t_2, t_4))\right]_+.$$ 
	In this formula, the infimum is computed at all times $t_2$ and $t_3$ satisfying the conditions $t_1 \leq t_2\leq t_3 \leq t_4$ ($t_1$ and $t_4$ are fixed), according to the network's topology.
	
	The end-to-end server for flow 1 for the network in Figure~\ref{fig:tree} is 
	$$S_{\mathrm{e2e}}(t_1, t_4) = \left[\inf_{\substack{t_1 \leq t_3\leq t_4,\\t_2\leq t_3}} (S_1(t_1, t_3) + S_2(t_2, t_3) + S_3(t_3, t_4) - A_2(t_2, t_4) - A_3(t_1, t_3))\right]_+.$$
	Here again, the infimum is computed for all $t_2$ and $t_3$, but the conditions change with the network topology: no lower bound is enforced for $t_2$ as server 2 has no predecessor, that is still upper bounded by $t_3$ and the new bounds for $t_3$ become $t_1\leq t_3\leq t_4$, as $(1, 3)$ is an edge of the network. 
\end{example}
\begin{proof}[Proof of Theorem~\ref{th:e2e-sc}]
	Consider server $j$. For all $t_j \leq t_{j^{\bu}}$ in the same backlogged period,\\ 
	$\sum_{i \in  j} [A^{(j^\bu)}_i(0, t_{j^\bu}) - A^{(j^\bu)}_i(0, t_j)] \geq S_j(t_j, t_{j^\bu}).$ 
	In particular, this formula is true when $t_j = start_j(t_{j^\bu})$, the start of the backlogged period of $t_{j^{\bu}}$ at server $j$. In that case  $A^{(j^\bu)}_i(0, t_j) =A^{(j)}_i(0, t_j)$, and
	$$\sum_{i \in  j} [A^{(j^\bu)}_i(0, t_{j^\bu}) - A^{(j)}_i(0, t_j)] \geq S_j(t_j, t_{j^\bu}).$$
	Summing over all servers $j$, we obtain 
	$$\sum_{j=1}^n \big(\sum_{i \in  j} [A^{(j^\bu)}_i(0, t_{j^\bu}) - A^{(j)}_i(0, t_j)]\big) \geq \sum_{j=1}^n S_j(t_j, t_{j^\bu}).$$
	Now, by exchanging the two sums on the left-hand side, most of the terms cancel, and (remind that $A_i^{(\pi_i(1))} = A_i$)
	$$\sum_{i=1}^m [A^{(\pi_i(\ell_i)^{\bu})}_i(0, t_{\pi_i(\ell_i)^{\bu}}) - A_i(0, t_{\pi_i(1)})] \geq \sum_{j=1}^n S_j(t_j, t_{j^\bu}).$$
	Keeping only $A_1^{(\pi_i(\ell_i)^{\bu})}(0, t_{\pi_i(\ell_i)^{\bu}})$ on the left-hand side of the inequality and using $A_i^{(\pi_i(\ell_i)^{\bu})}(0, \cdot) \leq A_i(0, \cdot)$, we obtain
	\begin{align*}
		A^{(\pi_1(\ell_1)^\bu)}_1(0, t_{\pi_1(\ell_1)^\bu}) & \geq A_1(0, t_{\pi_1(1)}) +  \sum_{j=1}^n S_j(t_j, t_{j^\bu}) - \sum_{i=2}^m [A^{(\pi_i(\ell_i)^{\bu})}_i(0, t_{\pi_i(\ell_i)^{\bu}}) - A_i(0, t_{\pi_i(1)})] \\ 
		& \geq A_1(0, t_{\pi_1(1)}) +  \sum_{j=1}^n S_j(t_j, t_{j^\bu}) - \sum_{i=2}^m A_i(t_{\pi_i(1)}, t_{\pi_i(\ell_i)^{\bu}}),
	\end{align*}
	We also have $\forall j\in\pi_1$, $A_1^{(j)}(0, t_{j}) = A_1^{(j^{\bu})}(0, t_{j}) \leq A_1^{(j^{\bu})}(0, t_{j^{\bu}})$. So $A^{(\pi_1(\ell)^\bu)}_1(0, t_{\pi_1(\ell_1)^\bu}) \geq A_1(0, t_{\pi_1(1)})$. 

	The final result follows by taking the minimum on all possible values of $t_j$, $j\notin \{\pi_1(1), \pi_1(\ell_1)^{\bu}\}$, and noticing that $\pi_1(\ell_1)^\bu = n+1$.
\end{proof}

This result is valid for any bivariate arrival and service processes.
The next subsection demonstrates its use in the analytic combinatorics framework. 
Before that, let us briefly discuss the power of Theorem~\ref{th:e2e-sc}: 
Looking at the end-to-end server calculation, it can be noted that all arrival and service processes appear in the formula as they originally enter the network. 
We point out that the tree generalization is crucial here, because it avoids the need for characterizing internal flows that may have become dependent after sharing a server. 
Hence, this calculation introduces \textit{no method-pertinent dependencies}
at all. 
Moreover,  the end-to-end server calculation is performed in one big step, in contrast to a sequential application of network calculus operations. 
This is already known to \textit{avoid a sequencing penalty} \cite{Beck16},
 yet completely avoiding it by performing all operations \textit{simultaneously} is only possible now with Theorem~\ref{th:e2e-sc}.

\vspace{-2mm}

\subsection{Bounding generating function of the end-to-end server}
\label{sec:series-pmoo}
	In this subsection, we assume that 
	\begin{itemize}
		\item[$(H_1)$] all processes and servers are mutually independent;
		\item[$(H_2)$] flow $A_i$ is $(\sigma_{A_i}, \rho_{A_i})$-constrained for all flow $i\in\{1, \ldots,  m\}$ and $S_j$ is $(\sigma_{S_j}, \rho_{S_j})$-constrained for all server $j \in\{1, \ldots, n\}$. 
	\end{itemize}
	We can now derive a service bounding generating function for the end-to-end dynamic server. 
	
	\begin{theorem}
		\label{th:e2e-gs-indep}
		Under hypothesis $(H_1)$ and $(H_2)$, the end-to-end service for flow 1 is bounded by the service bounding generating function
		$$F_{S_{\mathrm{e2e}}}(\theta, z) =  \frac{e^{\theta(\sum_{i=2}^m\sigma_{A_i}(\theta) + \sum_{j=1}^n \sigma_{S_j}(\theta))}}{ \prod_{j\notin \pi_1} (1- e^{ - \theta(\rho_{S_j}(\theta) - \sum_{i\in j} \rho_{A_i}(\theta))})} \prod_{j\in\pi_1} \frac{1}{1-e^{ - \theta(\rho_{S_j}(\theta)- \sum_{1\neq i\in j} \rho_{A_i}(\theta))}z}.$$
	\end{theorem}

\begin{proof}
	From Theorem~\ref{th:e2e-sc}, we can derive a bound of the Laplace transform of the service for flow 1. For the sake of concision, in the last line, 
	 we omit the function argument $ \theta $. 
	\begin{align}
		\E[e^{-\theta S_{\mathrm{e2e}}(t_{\pi_1(1)}, t_{n+1})}] &\leq \sum_{\forall j,~t_j\leq t_{j^\bu}} \E[e^{ -\theta(\sum_{j=1}^n S_j(t_j, t_{j^\bu}) - \sum_{i=2}^m A_i(t_{\pi_i(1)}, t_{\pi_i(\ell_i)^\bu}))}] \notag\\
		\label{eq:indep}
	&= \sum_{\forall j,~t_j\leq t_{j^\bu}} \prod_{j=1}^n\E[e^{ -\theta S_j(t_j, t_{j^\bu})}]  \prod_{i=2}^m\E[e^{ \theta A_i(t_{\pi_i(1)}, t_{\pi_i(\ell_i)^\bu})}]\\
		&\hspace{-3cm}\leq \sum_{\forall j,~t_j\leq t_{j^\bu}} \prod_{j=1}^n e^{\theta(\sigma_{S_j}(\theta) - \rho_{S_j}(\theta)(t_j^\bu - t_j))}\prod_{i=2}^m e^{\theta(\sigma_{A_i}(\theta) + \rho_{A_i}(\theta)(t_{\pi_i(\ell_i)^\bu} - t_{\pi_i(1)}))} \notag\\
	&	\hspace{-3cm}\leq \sum_{\substack{\forall j,~u_j \geq 0,\\ \sum_{j\in\pi_1} u_j= t_{n+1} - t_{\pi_1(1)}}}\prod_{j=1}^n e^{\theta(\sigma_{S_j}(\theta) - \rho_{S_j}(\theta)u_j)}\prod_{i=2}^m e^{\theta(\sigma_{A_i}(\theta) + \rho_{A_i}(\theta)(\sum_{j\in\pi_i} u_j))} \notag\\
	&	\hspace{-3cm}= e^{\theta(\sum_{i=2}^m\sigma_{A_i}(\theta) + \sum_{j=1}^n \sigma_{S_j}(\theta))} \sum_{\substack{\forall j,~u_j \geq 0,\\ \sum_{j\in\pi_1} u_j= t_{n+1} - t_{\pi_1(1)}}}\prod_{j=1}^n e^{ - \theta(\rho_{S_j}(\theta) - \sum_{1\neq i\in j} \rho_{A_i}(\theta))u_j} \notag\\
	& \hspace{-3cm} = \frac{e^{\theta(\sum_{i=2}^m\sigma_{A_i} + \sum_{j=1}^n \sigma_{S_j})}}{ \prod_{j\notin \pi_1} \big(1- e^{ - \theta(\rho_{S_j} - \sum_{i\in j} \rho_{A_i})}\big)}\sum_{\sum_{j\in\pi_1} u_j = t_{n+1} - t_{\pi_1(1)}} \prod_{j\in\pi_1} e^{ - \theta(\rho_{S_j}- \sum_{1\neq i\in j} \rho_{A_i})u_j}. \notag
	\end{align}
	In Equation~\eqref{eq:indep}, we use the independence of the arrivals and service; in the third line, the constraints on the processes; in the fourth line, we changed the variables: $u_j = t_j^\bu - t_j$ (hence $t_{\pi_i(\ell_i)^\bu } - t_{\pi_i(1)} = \sum_{j\in\pi_i} u_j$), 
	 and in the sixth line summed the terms not constrained by $\sum_{j\in\pi_1} u_j = t_{n+1} - t_{\pi_1(1)}$.

	We can recognize the terms of the Cauchy product of geometric series, hence the service is bounded by the generating function
	\begin{align*}
		F_{S_{\mathrm{e2e}}}(\theta, z) & =  \frac{e^{\theta(\sum_{i=2}^m\sigma_{A_i}(\theta) + \sum_{j=1}^n \sigma_{S_j}(\theta))}}{ \prod_{j\notin \pi_1} \big(1- e^{ - \theta(\rho_{S_j}(\theta) - \sum_{i\in j} \rho_{A_i}(\theta))}\big)} \prod_{j\in\pi_1} \frac{1}{1-e^{ - \theta(\rho_{S_j}(\theta)- \sum_{1\neq i\in j} \rho_{A_i}(\theta))}z}.
	\end{align*}
\end{proof}



\begin{example}
	A service bounding generating function for the end-to-end server for flow 1 in the network of Figure~\ref{fig:toy} is 
 $$F_{e2e}(\theta, z) = \frac{e^{\theta(\sigma_{A_2} + \sigma_{A_3} + \sigma_{S_1} + \sigma_{S_2} + \sigma_{S_3})} }{(1-e^{-\theta(\rho_{S_1} - \rho_{A_2})}z)(1-e^{-\theta(\rho_{S_2} - \rho_{A_2} - \rho_{A_3})}z)(1-e^{-\theta(\rho_{S_3} - \rho_{A_3})}z)}.$$




A service bounding generating function for the end-to-end server for flow 1 in the network of Figure~\ref{fig:tree} is 
$$F_{e2e}(\theta, z) = \frac{e^{\theta(\sigma_{A_2} + \sigma_{A_3} + \sigma_{S_1} + \sigma_{S_2} + \sigma_{S_3})}}{1-e^{-\theta(\rho_{S_2} - \rho_{A_2})}} \frac{1}{1-e^{-\theta(\rho_{S_1} - \rho_{A_3})}z} \frac{1}{1-e^{-\theta(\rho_{S_3} - \rho_{A_2})}z}.$$
Note the slight difference between the two formulas: informally, there is a '$z$' variable only for factors corresponding to the servers crossed by flow 1. 
\end{example}

The end-to-end service process is not necessarily $(\sigma,\rho)$-constrained. 
Corollary~\ref{cor:simple-pmoo} below gives a simpler  bounding generating function, a proof is given in Appendix~\ref{app:simple-pmoo}. 
In short, $G(\theta, z)$, defined in the statement of Corollary~\ref{cor:simple-pmoo}, is a bounding generating function of $F_{e2e}(\theta, z)$ obtained by bounding its factors not related to the dominant singularity.

 Let us first alleviate the notations and denote $$F_{S_{\mathrm{e2e}}}(\theta, z) = e^{\theta\sigma_{S_{\mathrm{e2e}}}(\theta)} \prod_{j\in\pi_1} \frac{1}{1-e^{-\theta \rho'_j(\theta)}z},$$
with 
$$
	e^{\theta\sigma_{S_{\mathrm{e2e}}}(\theta)} = \frac{e^{\theta(\sum_{i\neq 1}\sigma_{A_i}(\theta) + \sum_j \sigma_{S_j}(\theta))}}{ \prod_{j\notin \pi_1} \big(1- e^{ - \theta(\rho_{S_j}(\theta) - \sum_{i\in j} \rho_{A_i}(\theta))}\big)} \quad \text{ and } \quad 
	\rho'_j(\theta)  = \rho_{S_j}(\theta)- \sum_{1\neq i\in j} \rho_{A_i}(\theta),$$ 
%
and assume (without loss of generality, by renumbering the servers) that $\rho'_j(\theta) = \rho'_1(\theta)$ for all $j\leq k$ and  $\rho'_j(\theta) > \rho'_1(\theta)$ for all $j>k$. 

\begin{corollary}
	\label{cor:simple-pmoo}
	$G(\theta, z) = \frac{e^{\theta \sigma_{S_{\mathrm{e2e}}}(\theta)}}{\prod_{j=k+1}^m \big(1-e^{-\theta(\rho'_j(\theta) - \rho'_1(\theta))}\big)} \left( \frac{1}{1-e^{-\theta \rho'_1(\theta)}z}\right)^k$ is a bounding generating function of $F_{S_{\mathrm{e2e}}}(\theta, z)$.
\end{corollary}

	
\subsection{Performance bounds}
\label{sec:perf-pmoo}
We can now apply the results of Subsection~\ref{sec:snc-comb} to the end-to-end dynamic server. 
In the following, we assume stability, that is
we consider only values of $\theta$ such that for all $j\in\{1, \ldots, n\}$, $\rho'_j(\theta) > \rho_{A_1}(\theta)$ (the residual rate at server $ j $ is larger than the arrival rate).
\paragraph{Output departure process and backlog bound}
With Lemma~\ref{lem:output}, we can compute an arrival bounding generating function for the departure process: 
\begin{equation*}
F_D(\theta, z) = \frac{e^{\theta(\sigma_{A_1}(\theta) + \sigma_{S_{\mathrm{e2e}}}(\theta))} }{\prod_{j\in\pi_1}\big(1-e^{\theta (\rho_{A_1}(\theta) - \rho'_j(\theta))}\big)} \frac{1}{1-e^{\theta\rho_{A_1}(\theta)}z}.
\end{equation*}
We recognize a $(\sigma, \rho)$-constraint for the departure process, and observe that $\rho_D(\theta) = \rho_A(\theta)$. 

From Corollary~\ref{lem:backlog}, the backlog bound for flow 1 is then
	\begin{equation*}
	\p(q(t) \geq b) \leq \frac{e^{\theta(\sigma_{A_1}(\theta) + \sigma_{S_{\mathrm{e2e}}}(\theta))}}{ \prod_{j\in\pi_1}(1-e^{\theta (\rho_{A_1}(\theta) - \rho'_j(\theta))})} e^{-\theta b}.
\end{equation*}

\paragraph{Delay bound}
It is not as straightforward to compute the delay bound: the bounding generating function of the service does not correspond to a $(\sigma, \rho)$-constraint. The bounding generating function for the delay is 
	\begin{equation}
		\label{eq:e2e-delay}
		F_d(\theta, z) = e^{\theta \sigma_{A_1}(\theta)} \frac{e^{\theta\rho_{A_1}(\theta)}F_{S_{\mathrm{e2e}}}(\theta, e^{\theta\rho_{A_1}(\theta)}) -  zF_{S_{\mathrm{e2e}}}(\theta, z)}{1- e^{-\theta\rho_{A_1}(\theta)}z}.
	\end{equation}

As $F_{S_{\mathrm{e2e}}}(\theta, z)$ is a rational function, $F_d(\theta, z)$ is. Moreover, all the singularities of these functions are known. 
Note that $e^{\theta \rho_{A_1}(\theta)}$ is not a singularity of $F_d(\theta, z)$. Indeed,  when $z = e^{-\theta\rho_{A_1}(\theta)}$, the numerator of Equation~\eqref{eq:e2e-delay} is zero. 
The singularities of $F_d(\theta, z)$ are then those of $F_{S_{\mathrm{e2e}}}(\theta, z)$, that is $e^{\theta\rho'_j(\theta)}$, for $j\in\{1,\ldots, n\}$. 
All the terms of the generating function of Equation~\eqref{eq:e2e-delay} can be computed automatically with symbolic computation tools such as SageMath~\cite{sagemath}. 
However, there are several cases where the exact expression is compact enough to be given here.

\begin{itemize}
 \item Residual rates are all distinct (see derivation in Appendix~\ref{app:sing1}): $\forall T\geq 0$,  
	\begin{multline}
		\label{eq:sing1}
		\p(d(t) \geq T)  \leq  [z^T]F_{d}(\theta, z) = \sum_{j=1}^n \frac{e^{\theta(\sigma_{A_1}(\theta)  + \sigma_{S_{\mathrm{e2e}}}(\theta)  + \rho_{A_1}(\theta))}}{1-e^{\theta(\rho_{A_1}(\theta) - \rho'_j(\theta))}}
		  \cdot  \left[\prod_{k\neq j} \frac{1}{1-e^{\theta(\rho'_j(\theta) - \rho'_k(\theta))}}\right] e^{-\theta\rho'_j(\theta)T}.
	\end{multline}
	\item Residual rates are all equal (see derivation in Appendix~\ref{app:uniform}):  $\forall T >  0$,
	\begin{multline}
		\hspace{-4mm}
		\p(d(t) \geq T)  \leq  [z^T]F_{d}(\theta, z)
		=  e^{\theta(\sigma_{A_1}(\theta) + \sigma_{S_{\mathrm{e2e}}}(\theta) + \rho_{A_1}(\theta))}
			  \cdot \left[\sum_{i=1}^{n}\binom{T + i-2}{T-1} \frac{e^{-\theta\rho'_{1}(\theta)T}}{\big(1-e^{-\theta(\rho'_1(\theta) - \rho_{A_1}(\theta))}\big)^{n-i+1}}  \right]. \label{eq:uniform}
	\end{multline}
\end{itemize}
In the general case, one can first use the service bounding generating function of Corollary~\ref{cor:simple-pmoo} to be in the case of equal residual rates, and then use Equation~\eqref{eq:uniform} to obtain $\forall T >  0$,
\begin{multline}
	\label{eq:bound}
\p(d(t) \geq T) \leq \frac{e^{\theta \left(\sigma_{A_1}(\theta) + \sigma_{S_{\mathrm{e2e}}}(\theta) + \rho_{A_1}(\theta)\right)}}{\prod_{j=k+1}^n{1-e^{-\theta (\rho'_j(\theta) - \rho'_1(\theta))}}} 
 \cdot\left[
\sum_{i=1}^{k} \binom{T + i-2}{T - 1} \frac{e^{-\theta \rho'_1(\theta)T}}{\big(1-e^{-\theta(\rho'_1(\theta) - \rho_{A_1}(\theta))}\big)^{k-i + 1}}\right].
\end{multline}

\subsection{Recovering existing results for the canonical tandem}

\begin{figure}[t]
	\centering
	\resizebox{0.62\textwidth}{!}{%
	\begin{tikzpicture}[server/.style={shape=rectangle,draw,minimum height=.8cm,inner xsep=3ex}]
		\node[server,name=S1] at (0,0) {1};
		\node[server,name=S2] at (2,0) {2};
		\node[server,name=S3] at (4,0) {3};
		\node[server,name=S3] at (7,0) {$n$};
		\draw[->, thick, red] (-1,0) node[left] {1} -- (8,0);
		\draw[->, thick, black] (-1,0.4) node[left] {2} to[out=-10,in=-170] (1,.4);
		\draw[->, thick, black]  (1.2,-0.4) node[left] {3} to[out=10,in=170] (3,-.4);
		\draw[->, thick, black]  (3.2,0.4) node[left] {4} to[out=-10,in=-170] (5,.4);
		\draw[thick, black, dotted]  (5,-0.2) -- (5.5,-.2);
		\draw[->, thick, black]  (6,-0.4) to[out=10,in=170] (8,-.4) node[right] {$ n + 1 $};
	\end{tikzpicture}
}%
	\vspace{-2mm}
	\caption{Canonical tandem\label{fig:canonical-tandem}.}
	\vspace{-2mm}
\end{figure}
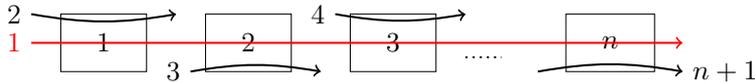

In this subsection, we recover the state of the art in SNC for the ``canonical tandem'' (Figure~\ref{fig:canonical-tandem}), as it is frequently analysed in the literature, e.g. \cite{CBL06, Fidler06, RF11}. 
Since this simple flow interference structure provides no opportunity for the PMOO principle to pay off, our new bounds based on analytic combinatorics cannot further improve it (in contrast to the scenarios we present in Section~\ref{sec:numerical}).
To the best of our knowledge, the best technique for this topology is to calculate the optimum of two bounds originally derived in \cite[Theorem~3]{Fidler06}, denoted by ``\cite{Fidler06}~Thm.~3'' below.
However, this bound targets a closed-form solution for the stochastic delay and, hence,  compromises on tightness.
Attempting a fair comparison, we also compare the novel result from  Equation~\eqref{eq:bound} to a tighter version of ``\cite{Fidler06}~Thm.~3'' that does not make this compromise and is given in 
\cite{Fidler06}~Eqn.~(9) as an intermediate result.
One can show that our Equation~\eqref{eq:bound} recovers the result in \cite{Fidler06}~Eqn.~(9), since both start with the same inequalities and, afterwards, only equivalent term manipulations are used.

\begin{figure*}[t]
	\centering
	\subfloat[\centering util~=~0.55]{
			\includegraphics[width=0.3\textwidth]{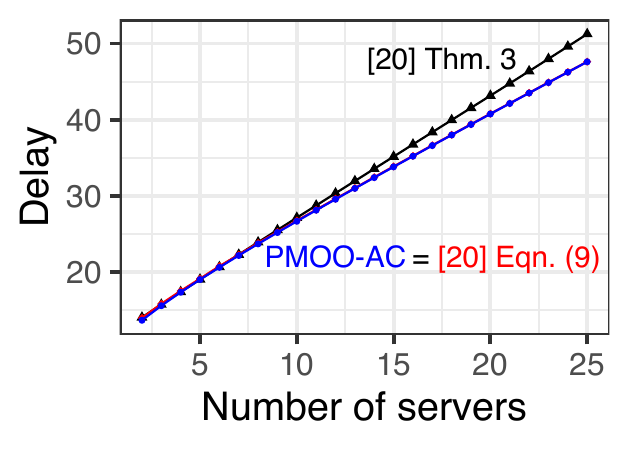}
		}
	\hspace{5mm}
	\subfloat[\centering util~=~0.9]{
			\includegraphics[width=0.3\textwidth]{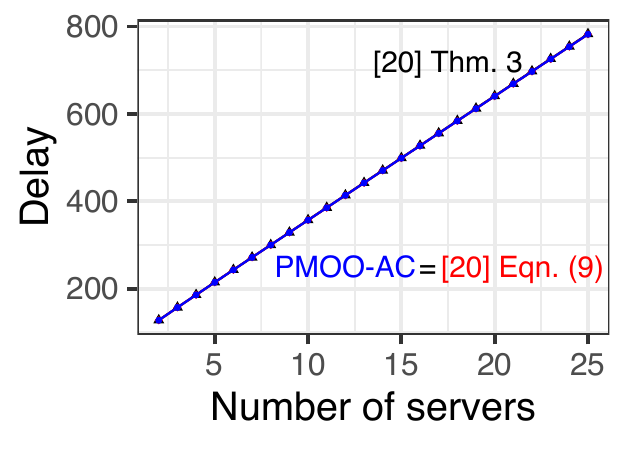}
		}
	\vspace{-2mm}
	\caption{Stochastic delay bounds for the canonical tandem (for traffic with independent exponentially distributed increments with $ \lambda=1 $ and constant rate servers). \label{fig:canonical tandem-exp}}
	\vspace{-4mm}
\end{figure*}

We also verified this numerically for exponential i.i.d. arrivals in Figure~\ref{fig:canonical tandem-exp}.
The closed-form solution in ``\cite{Fidler06}~Thm.~3'' tends to lead to less accurate bounds, depending on the utilization as well as the number of servers.
We note that it is the only approach that exploits the positive part in Theorem~\ref{thm:residual-ser}, but this effect is negligible for the chosen traffic class.
Anyway, the positive part can also be incorporated in the analytic combinatorics; yet, we omit this here for the sake of brevity.

\vspace{-2mm}

\subsection{Relation to state-of-the-art analyses}
\label{subsec:sota-analysis}

A typical network calculus analysis to derive end-to-end performance bounds consists of three separate sequential steps:
\begin{enumerate}
	\item Reducing the network to a tandem traversed by the 
	flow of interest, by output bounding.
	\item Reducing the tandem to an end-to-end server that represents the complete system.
	\item Computation of performance bounds.
\end{enumerate}

Various SNC analysis techniques have been derived in the literature.
In \cite{BS13-1}, a \textit{sequential} separated flow analysis (SFA) is used, where each network calculus operation is applied one after the other in contrast to a simultaneous technique as in Theorem~\ref{th:e2e-sc}.
On the other hand, it has been shown in \cite{Fidler06} that combining step (2) and (3) is able to avoid a sequencing penalty to some degree (see also \cite{Beck16}).
Moreover, one could also convolve servers as much as possible before subtracting cross-flows. 
In \cite{NS20}, this is called ``PMOO'', however, it is not fully able to use the power of the PMOO principle, as it involves a sequential order of network calculus operations which is why we call it ``seqPMOO'' in the following. 

However, none of the above mentioned techniques is able to avoid a sequencing penalty entirely.
On the other hand, the newly introduced PMOO analysis avoids this pitfall by combining all three steps of the analysis and also reduces the overall number of calculations.
In the next section (Section~\ref{sec:numerical}), we show that this approach leads to significantly more accurate delay bounds.

\vspace{-1mm}

	\section{Numerical Evaluation of Tree Networks}
\label{sec:numerical}

In this section, we compute bounds on the delay's violation probability and stochastic delay bounds applying state-of-the-art techniques as well as our unleashed PMOO. We perform several experiments for different network topologies. 
For the arrivals, we assume three discrete-time processes all adhering to the class of $(\sigma_A, \rho_A)$-constrained arrivals: 
Exponentially distributed arrivals increments \cite{Beck16} with parameter $ \lambda $,
Weibull distributed arrival increments with fixed shape parameter $ k=2 $ and scale parameter $ \lambda $,
and discrete-time Markov-modulated On-Off (MMOO) arrivals \cite{Chang00, Beck16}.
The latter can be described by three parameters: the probability to stay in the ``On''-state in the next time step, $ p_{\mathrm{on}} $, the probability to stay in the ``Off''-state, $ p_{\mathrm{off}} $, and a constant peak rate $ P$, at which data is sent during the ``On''-state.
For the service, we always assume work-conserving $ S $-servers with a constant rate.

\subsection{Interleaved tandem}

\begin{figure*}[b]
	\vspace{-5mm}
	\centering
	\subfloat[\centering Exponential distribution with~$\lambda_i = 1.5$, $ i = 1,2,3 $ \label{fig:interleaved-tandem-intro-again}]{
		\includegraphics[width=0.3\textwidth]{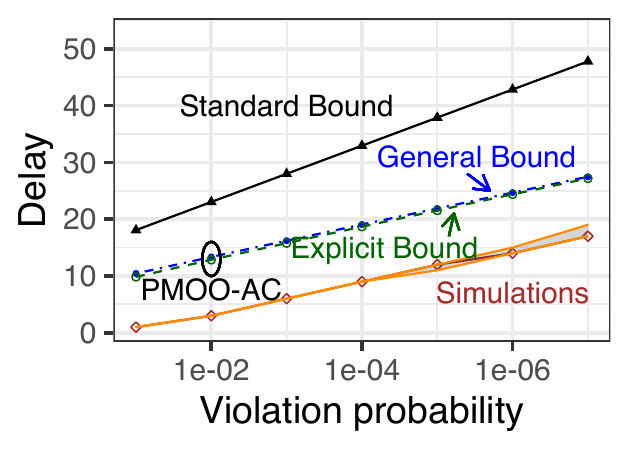}
	}
	\subfloat[\centering Weibull distribution with~$\lambda_i = 1.0, $ $ i = 1,2,3 $]{
		\includegraphics[width=0.3\textwidth]{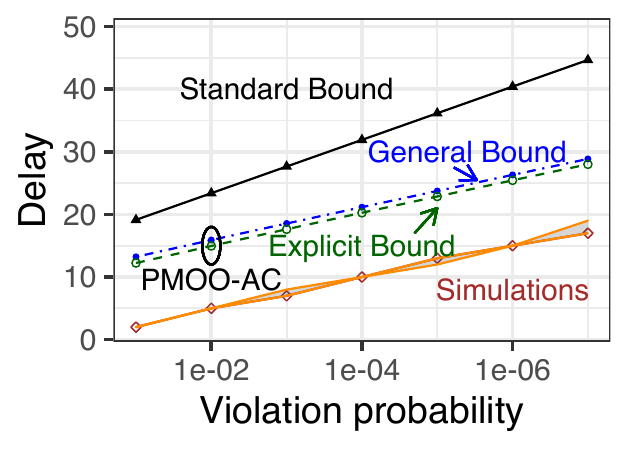}
	}
	\subfloat[\centering MMOO with $ p_{\mathrm{on}, i} = 0.5 $, $ p_{\mathrm{off}, i} = 0.5,$ $ P_i = 1.4,$ $ i = 1,2,3 $]{
		\includegraphics[width=0.3\textwidth]{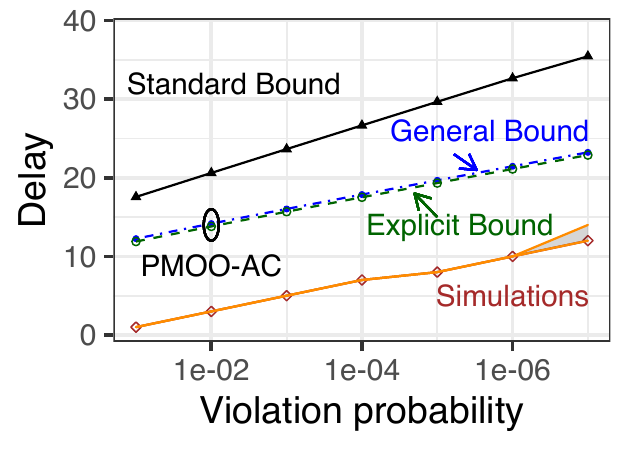}
	}
	\vspace{-2mm}
	\caption{Delay bounds and simulations for the interleaved tandem with server rates $ C_1 = 2.5, $  $C_2 = 3.0, $ $ C_3 = 2.0. $}
	\label{fig:interleaved-tandem-results}
\end{figure*}

In our first numerical evaluation, we calculate stochastic delay bounds for a standard example in network calculus when PMOO effects shall be illustrated: the interleaved tandem network in Figures~\ref{fig:toy}.
These are compared to simulation results.
The ``standard bound'' consists of the minimum of the state-of-the-art techniques described in Section~\ref{subsec:sota-analysis}.
They are then compared to our new analytic combinatorics based PMOO using two variants:
The first, again based on Equation~\eqref{eq:bound}, is called ``PMOO-AC: General Bound'' in this example.
The second, ``PMOO-AC: Explicit Bound'' exploits that we choose all parameters such that the residual rates are distinct.
Therefore, we can apply Equation~\eqref{eq:sing1}.
In addition, we provide simulation results with respective pointwise 95\%-confidence bands that are based on order statistics of a binomially distributed sample \cite{HM11} to see how the bounds relate to the empirical delay.
The results are provided in Figure~\ref{fig:interleaved-tandem-results}.

We observe that, while the PMOO-AC bounds lead to quite similar results, they both outperform the standard bound significantly. 
For example, for a delay violation probability of $10^{-3}$, the delay bound is improved from 28 to 18, and for a violation probability of $10^{-7}$, from 45 to 31. 
These examples indicate an improvement of more than 35\%.


These positive results are mainly caused by the fact that the PMOO analysis is able to provide bounds without introducing any method-pertinent dependencies.
The standard bound, in contrast, suffers from such a dependency in the calculation and, consequently, needs to apply Hölder's inequality.
Furthermore, it also potentially loses accuracy due to the sequencing penalty.

In addition, we see that the gap between simulation results and bounds is considerably reduced and 
the scaling of the delay in the simulations
is captured well (getting closer to single-node results again). Recalling that  Figure~\ref{fig:interleaved-tandem-intro-again} without the new bounds has been used in Section~\ref{sec:intro} to illustrate that the known SNC gap will widen too much, we now provide a new prospect for this again.

 
 \vspace{-2mm}

\subsection{Extended interleaved tandem}
\label{subsec:extend-interleaved-tandem}

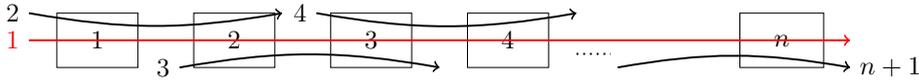
\begin{figure}[t]
	\centering
	\resizebox{0.75\textwidth}{!}{%
	\begin{tikzpicture}[server/.style={shape=rectangle,draw,minimum height=.8cm,inner xsep=3ex}]
		\node[server,name=S1] at (0,0) {1};
		\node[server,name=S2] at (2,0) {2};
		\node[server,name=S3] at (4,0) {3};
		\node[server,name=S3] at (6,0) {4};
		\node[server,name=S3] at (10,0) {$n$};
		\draw[->, thick, red] (-1,0) node[left] {1} -- (11,0);
		\draw[->, thick, black] (-1,0.4) node[left] {2} to[out=-10,in=-170] (2.7,.4);
		\draw[->, thick, black]  (1.2,-0.4) node[left] {3} to[out=10,in=170] (5,-.4);
		\draw[->, thick, black]  (3.2,0.4) node[left] {4} to[out=-10,in=-170] (7,.4);
		\draw[thick, black, dotted]  (7,-0.2) -- (7.5,-.2);
		\draw[->, thick, black]  (7.6,-0.4) to[out=10,in=170] (11,-.4) node[right] {$ n + 1 $};
	\end{tikzpicture}
	}%
	\vspace{-2mm}
	\caption{Extended interleaved tandem.} \label{fig:extended-interleaved-tandem}
	\vspace{-3mm}
\end{figure}

In the next experiment, we generalize the case of an interleaved tandem by varying its lengths while keeping the interference structure (Figure~\ref{fig:extended-interleaved-tandem}).
Here, the standard bound only includes the sequential and simultaneous SFA approaches, as the number of possible (combinations of) network calculus operations in the seqPMOO grows exponentially with the number of servers and, therefore, is computationally prohibitive.


In Figure~\ref{fig:extended-interleaved-tandem-results}, we show stochastic delay bounds for a fixed violation probability of $ 10^{-6}. $ 
While the standard bound explodes in the number of servers (we have only included the results from 3 to 5 servers), the new technique scales significantly better.
This is mainly due to the fact that the SFA leads to the application of $ n - 1 $ Hölder inequalities for an interleaved tandem of size $ n $, whereas the new PMOO-AC does not incur any method-pertinent dependencies in the analysis of this topology.

\begin{figure*}[t]
	\centering
	\subfloat[\centering Exponential distribution with~$~\lambda_i = 2.0 $]{
		\includegraphics[width=0.3\textwidth]{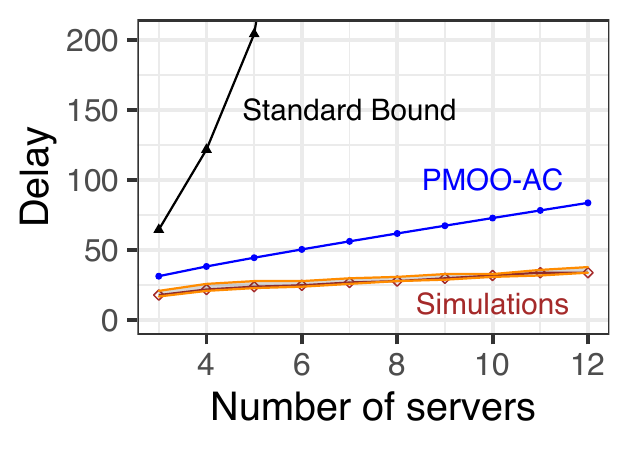}
	}
	\subfloat[\centering Weibull distribution with~$\lambda_i~=~0.7 $]{
		\includegraphics[width=0.3\textwidth]{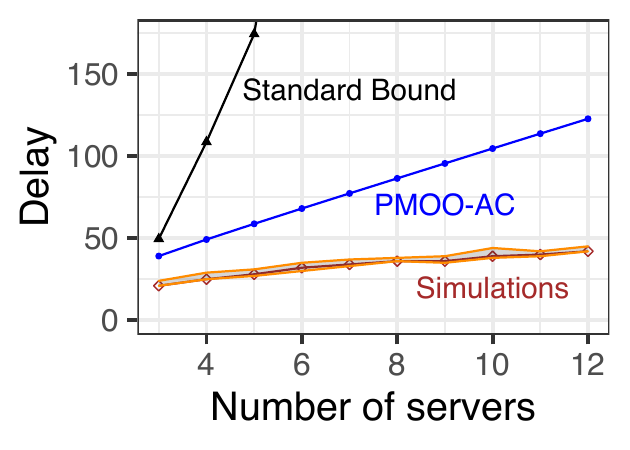}
	}
	\subfloat[\centering MMOO with $ p_{\mathrm{on}, i} = 0.5, $  $p_{\mathrm{off}, i} = 0.5, $ $ P_i = 1.0 $]{
		\includegraphics[width=0.3\textwidth]{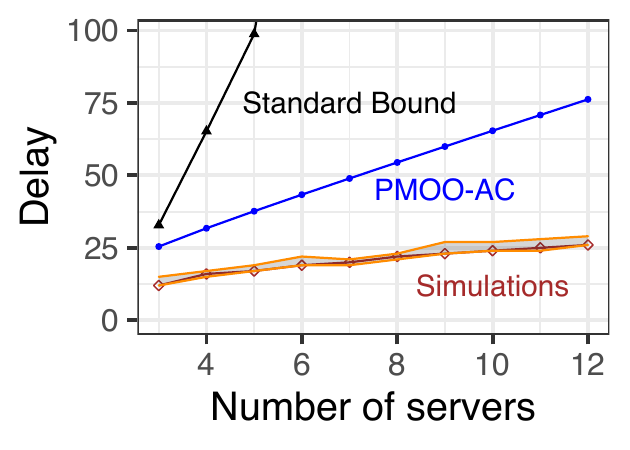}
	}
	\vspace{-1mm}
	\caption{Delay bounds for the extended interleaved tandem with server rates $ C_i = 2.0 $ for $ i= 1, \dots, 12. $
		\label{fig:extended-interleaved-tandem-results}}
	\vspace{-4mm}
\end{figure*}

Not only does this lead to tighter performance bounds, it also impacts runtimes significantly as it reduces an $ n $-dimensional non-linear optimization problem (one $ \theta $ and $ n-1 $ Hölder parameters) to a 1-dimensional optimization problem.\footnote{In our numerical experiments, we apply a grid search followed by a downhill simplex algorithm to optimise the parameters.}
Specifically, for the standard bound (in the case of MMOO arrivals), runtimes increase quickly from 5.2 seconds (3 servers) over 2 minutes (4 servers) to roughly 1.5 hours (5 servers).
On the other hand, PMOO-AC took a maximum runtime of 0.46 seconds in the case of 12 servers.

\vspace{-2mm}

\subsection{Case study: tree network}


\begin{figure}[b]
	\vspace{-4mm}
	\centering
	\resizebox{0.38\textwidth}{!}{%
	\begin{tikzpicture}[server/.style={shape=rectangle,draw,minimum height=.8cm,inner xsep=3ex}]
		\node[server,name=S1] at (0,0) {1};
		\node[server,name=S2] at (0,2) {2};
		\node[server,name=S3] at (2,0) {3};
		\node[server,name=S4] at (4,0) {4};
		\draw[->, thick, red] (-1,-0.12) node[left] {1} -- (5,-0.12);
		\draw[->, thick, black] (-1,-0.3) node[below left] {2} -- (1,-0.3);
		\draw[->, thick, black] (-1,2.15) node[left] {3} -- (1.1,2.15) -- (1.1, 0.3) -- (3, 0.3);
		\draw[->, thick, black]  (-1,1.85) node[left] {4}  -- (0.9,1.85) -- (0.9,0.15) -- (5,0.15);
	\end{tikzpicture}
	}%
	\vspace{-3mm}
	\caption{Case study tree network.} \label{fig:case-for-tree-networks}
\end{figure}
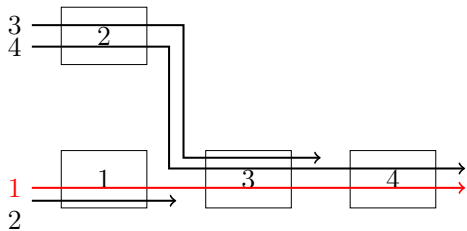


In Section~\ref{subsec:sota-analysis}, we explained that the PMOO-AC calculates performance bounds by combining all the steps of the analysis into a single large one.
The following case study shall mainly investigate the benefit of this.

To that end, we consider the tree network in Figure~\ref{fig:case-for-tree-networks}.
Standard SNC techniques first compute the output bounds of flows 3 and 4 at server 2 in order to compute the residual service at servers 3 and 4; this incurs again method-pertinent dependencies.
The application of Theorem~\ref{th:e2e-sc} allows us to circumvent these dependencies. 
In contrast, state-of-the-art analysis using SFA needs three applications of Hölder's inequality and seqPMOO requires one, respectively.

\begin{figure*}[t]
	\centering
	\subfloat[\centering Exponential distribution with~$~\lambda_i = 2.0$]{
		\includegraphics[width=0.3\textwidth]{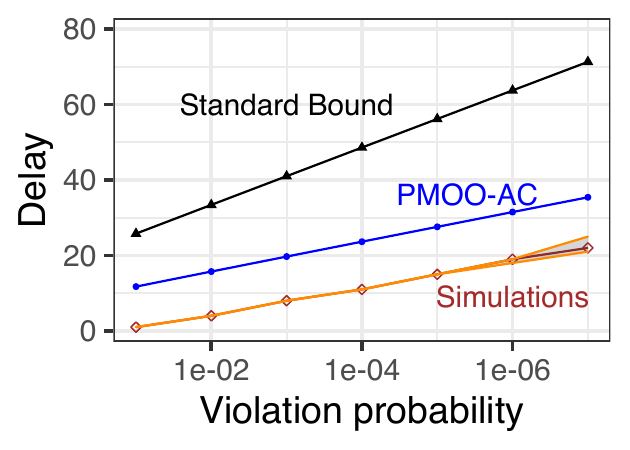}
	}
	\subfloat[\centering Weibull distribution with $ \lambda_i~=~0.67$]{
		\includegraphics[width=0.3\textwidth]{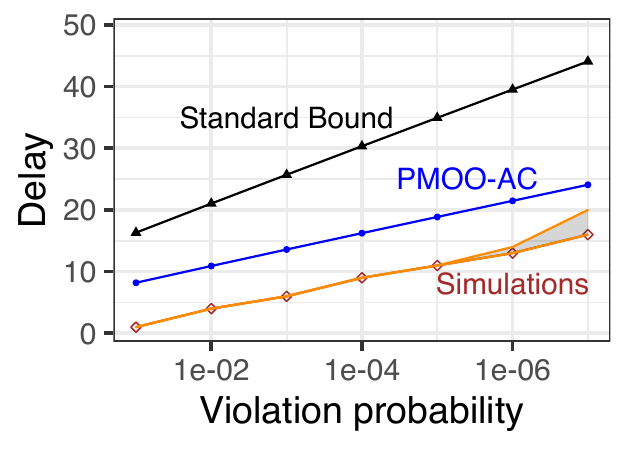}
	}
	\subfloat[\centering MMOO with $ p_{\mathrm{on}, i} = 0.5, $  $ p_{\mathrm{off}, i} = 0.5, $ $ P_i = 1.0 $]{
		\includegraphics[width=0.3\textwidth]{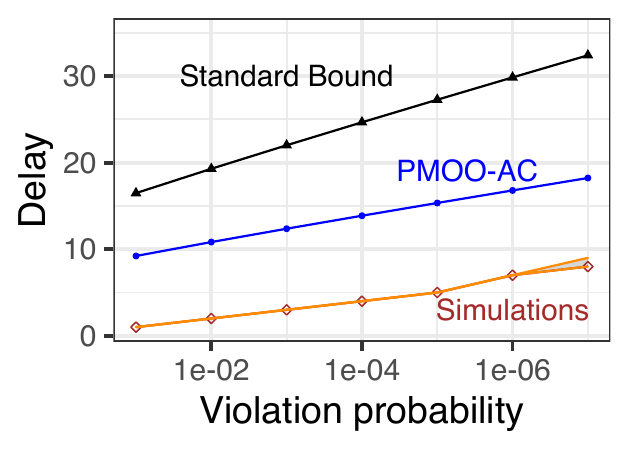}
	}	
	\vspace{-2mm}
	\caption{Delay bounds and simulations for the tree network with server rates $ C_i = 2.0, $ $ i = 1, \dots, 4. $ \label{fig:case-for-tree-networks-results}}
	\vspace{-3mm}
\end{figure*}

The results are shown in Figure~\ref{fig:case-for-tree-networks-results}.
Similar to the results for the interleaved tandem, we observe that the PMOO considerably improves the bound on the delay's violation probability.
Again, we are able to achieve a similar scaling compared to the simulation results, in contrast to the standard bound.
Further, while the state-of-the-art analysis requires an optimisation of up to 4 parameters (three Hölder and $ \theta $), the PMOO only has $ \theta $ to optimise, 
substantially improving the runtime. 


	\section{Towards Feedforward Network Analysis under Partial Dependence}
\label{sec:dependent}
In previous sections, we focused on mutually independent flows and servers in tree-reducible networks. 
However, if we want to analyse general feedforward networks, rejoining flows eventually force us to deal with dependencies. As discussed in Section~\ref{sec:pmoo}, in the presence of rejoining flows, 
the network will not be tree-reducible any longer. Thus, we cannot completely avoid method-pertinent dependencies using PMOO, but need to extend it to the partially dependent case.

As mentioned in Section~\ref{sec:framework}, the use of Hölder's inequality allows us to calculate performance bounds also in the dependent case.
Yet, as we discuss in this section, its number of applications as well as their induced inaccuracy can be influenced by the transformation to a tree-reducible network -- our goal is to minimize the impact of Hölder's inequality on the performance bounds.

First, in Subsection~\ref{sec:dep-curves}, we extend PMOO under partial dependence; next, in Subsection~\ref{sec:diamond}, we apply this to the canonical example of rejoining flows: the ``diamond network''. 
Finally, in Subsection~\ref{sec:discussion}, we discuss how to extend these promising results to larger feedforward networks.

\subsection{End-to-end server and performance bounds}
	\label{sec:dep-curves}

	To deal with the partial dependency of random variables, we adapt the notion of \textit{dependency graph} defined to prove the Lov\'asz Local Lemma (see~\cite{Mitz05}, for instance). Let $(X_h)_{h\in H}$ be a finite family of processes. In our case a process can be a service process $S_j$ or an arrival process $A_i$. 
	The dependency graph associated to $(X_h)_{h\in H}$ is the graph whose vertices are the processes $(X_h)_{h\in H}$ and process $X_h$ is not adjacent with processes with which it is mutually independent.
	More precisely, if $\Gamma(X_h)$ is the set of neighbours of $X_h$, $X_h$ is mutually independent of  $\{X_{h'}~|~X_{h'}\notin(\Gamma(X_h)\cup \{X_h\})\}$.

	Assume that the dependency graph contains $K$ strongly connected components $G_1,\ldots, G_K$, and that for all $h\in H$, and for some $\theta > 0$,  $\E[e^{\theta X_h(s_h, t_h)}] \leq e^{\theta\lambda_h(\theta, s_h, t_h)}$. The following lemma shows how to use as few H\"older inequalities as possible. 
	
	\begin{lemma}
		\label{lem:partialH}
		With the hypotheses and notations above, for positive $(p_h)_{h\in H}$ such that $ \forall k \in\{1, \ldots, K\} $,  $\sum_{h\in G_k} \frac{1}{p_h} = 1,$
		$$\E[e^{\theta(\sum_{h\in H} X_h(s_h, t_h))}] \leq \prod_{h\in H} e^{\theta\lambda_h(p_h\theta, s_h, t_h)}. $$
	\end{lemma}
	\begin{proof}
		\begin{multline*}
			\E[\prod_{h \in H} e^{\theta X_h(s_h, t_h)}]  =  \prod_{k=1}^K \E[\prod_{h \in G_k} e^{\theta X_h(s_h, t_h)}] 
			\leq \prod_{k=1}^K \prod_{h \in G_k}  (e^{p_{h}\theta \lambda_h(p_{h} \theta, s_h, t_h)}) ^{1/p_{h}} 
			= \prod_{h \in H} e^{\theta\lambda_h(p_h\theta, s_h, t_h)}.
		\end{multline*}
		The first equality is deduced from the dependency graph: we have mutual independence between families of processes of distinct connected components. 
		The inequality is the application of the H\"older inequality inside the connected components, as well as the use of the MGF bounds. 
	\end{proof}
	
We can now derive an end-to-end service for partially dependent servers and arrival processes. 
%
\begin{theorem}
	\label{th:e2e-gs-dep}
	The end-to-end service for flow 1 is bounded by the service bounding generating function
	\begin{multline*}F_{S_{\mathrm{e2e}}} (\theta, (p_i), (q_j), z)=  \frac{e^{\theta(\sum_{i\neq 1}\sigma_{A_i}(p_i\theta) + \sum_j \sigma_{S_j}(q_j\theta))}}{ \prod_{j\notin \pi_1} (1- e^{ - \theta(\rho_{S_j}(q_j\theta) - \sum_{i\in j} \rho_{A_i}(p_i\theta))})} 
		\prod_{j\in\pi_1} \frac{1}{1-e^{ - \theta(\rho_{S_j}(q_j\theta)- \sum_{1\neq i\in j} \rho_{A_i}(p_i\theta))}z},
	\end{multline*}
	for all $(p_i), (q_j)$ such that for all $k\leq K$, $\sum_{A_i\in G_k\setminus \{1\}} \frac{1}{p_i} + \sum_{S_j\in G_k} \frac{1}{q_j}=1$.
\end{theorem}
\begin{proof}
	The proof follows almost along the same lines as in Theorem~\ref{th:e2e-gs-indep}. 
	Lemma~\ref{lem:partialH} is used at Equation~\eqref{eq:indep}.  
	If  $X_h = A_i$, we have $s_h = t_{\pi_i(1)}$, $t_h = t_{\pi_i(\ell_i)}$ and $\lambda_h(\theta, s_h, t_h) = \sigma_{A_i}(\theta) + \rho_{A_i}(\theta)(t_h-s_h)$. If  $X_h = S_j$, we have $s_h = t_{j}$, $t_h = t_{j^{\bu}}$ and $\lambda_h(\theta, s_h, t_h) = \sigma_{S_j}(\theta) - \rho_{S_j}(\theta)(t_h-s_h).$ 
\end{proof}

This is straightforward if $A_1$'s connected component is a singleton: $A_1$ is then mutually independent of all other processes and Lemmas~\ref{lem:output} and~\ref{lem:delay} apply. 
If $A_1$ is not mutually independent of all other processes, the end-to-end server and $A_1$ are not independent. 
While a first solution is to apply again H\"older inequality, we propose to compute the bound directly. 
Instead of computing the end-to-end server first, we directly start with the delay expression from Theorem~\ref{th:e2e-sc}. 
The proofs use the same arguments as before, so we defer them to Appendix~\ref{app:dependent}.
\begin{theorem}
	\label{th:dep}
	With the same notations as above,
	for all $(p_i), (q_j)$ such that for all $k\leq K$, $\sum_{A_i\in G_k} \frac{1}{p_i} + \sum_{S_j\in G_k} \frac{1}{q_j}=1$ and $\theta \geq 0$ such that for all $j\in\{1, \ldots, n\}$, $ \sum_{i\in j} \rho_{A_i}(p_i\theta) < \rho_{S_j}(q_j\theta)$,
	\begin{enumerate}
		\item the departure process of flow 1 is bounded by the arrival bounding generating function
		\begin{align*}
			F_{D_{\mathrm{e2e}}}(\theta, (p_i),(q_j), z) & = \frac{e^{\theta \sigma_{A_1}(p_1\theta)} F_{S_{\mathrm{e2e}}}(\theta, (p_i)_{i\neq 1}, (q_j), e^{\theta \rho_{A_1}(p_1\theta)})}{1-e^{\theta \rho_{A_1}(p_1\theta)}z};
		\end{align*}
		\item the delay of flow $A_1$ is bounded by the generating function
		\begin{multline*}F_{d}(\theta, (p_i), (q_j), z) = \frac{{e^{\theta (\sigma_{A_1}(p_1\theta) + \rho_{A_1}(p_1\theta))}} F_{S_{\mathrm{e2e}}}(\theta, (p_i)_{i\neq 1}, (q_j), e^{\theta \rho_{A_1}(p_1\theta)})}{1-e^{-\theta \rho_{A_1}(p_1\theta)}z} \\ - \frac{ e^{\theta \sigma_{A_1}(p_1\theta)}z F_{S_{\mathrm{e2e}}}(\theta, (p_i)_{i\neq 1}, (q_j), z)}{1-e^{-\theta \rho_{A_1}(p_1\theta)}z}.
		\end{multline*}
	\end{enumerate}
\end{theorem}

\subsection{Case study: diamond network}

\label{sec:diamond}

Here, we consider the diamond network from Figure~\ref{fig:toy-ff}, where two originally independent flows separate and then rejoin. 
If we assume some kind of resource-sharing scheduling policy at server 0, then their outputs are dependent when interfering at server~3.
This network is not tree-reducible.
	
\subsubsection{Transformation into a tree-reducible network}
\label{subsec:ff}

A solution to transform this feedforward network into a tree-reducible network is the \textit{cutting} of flows. 
If we are interested in the performance of flow 1, we can cut flow~2, either between servers~0 and 2 or between servers~2 and 3, as shown in Figure~\ref{fig:ff-cut}. Clearly, we obtain a tree-reducible network and can now apply Theorems \ref{th:e2e-gs-dep} and \ref{th:dep} from the previous subsection. 

The analysis of the two obtained networks follows similar lines, so let us focus on the transformation when flow~2 is cut between servers~2 and 3 (on the right of Figure~\ref{fig:ff-cut}). A flow $2''$ is then created, and it stochastically depends on flow~1, flow~2, server~0 and server~2. 
Using the independence of these flows and servers and results from Section~\ref{sec:perf-pmoo},  we can compute a bounding generating function for the arrival process of flow $2''$:  $$F_{A_{2''}} (\theta, z) = F_{D_2}(\theta,z) = \frac{e^{\theta(\sigma_{A_1}(\theta) + \sigma_{A_2}(\theta) +  \sigma_{S_0}(\theta) + \sigma_{S_2}(\theta))}}{(1- e^{-\theta \left(\rho_{S_0}(\theta) - \rho_{A_1}(\theta) - \rho_{A_2}(\theta)\right)}) (1- e^{-\theta \left(\rho_{S_2}(\theta) - \rho_{A_2}(\theta)\right)})} \frac{1}{1-e^{\theta\rho_{A_2}(\theta)}z}.$$

\begin{figure}[tb]
	\centering
	\resizebox{0.4\textwidth}{!}{%
		\begin{tikzpicture}
			[server/.style={shape=rectangle,draw,minimum height=.8cm,inner xsep=3ex}]
			\node[server,name=S0] at (0,0) {0};
			\node[server,name=S1] at (2,-1) {1};
			\node[server,name=S2] at (2,1) {2};
			\node[server,name=S3] at (4,0) {3};	
			\draw[thick, ->, red] (-1, -0.2) node [left] {$1$}-- (0, -0.2)-- (2, -1.2) -- (4, -0.2) -- (5, -0.2);
			\draw[thick, ->, blue] (-1, 0.2)  node [left] {$2$} -- (0, 0.2)-- (2, 1.2) -- (4, 0.2) -- (5, 0.2);
		\end{tikzpicture}
	}%
	\caption{Diamond network.}
	\label{fig:toy-ff}
\end{figure}
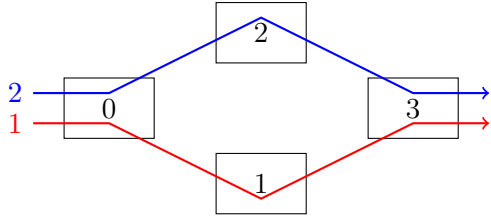

\begin{figure}[t]
	\centering
	\begin{minipage}{0.48\textwidth}
		\centering
		\begin{tikzpicture}
			[server/.style={shape=rectangle,draw,minimum height=.8cm,inner xsep=3ex}]
			\node[server,name=S0] at (0,0) {0};
			\node[server,name=S1] at (2,-1) {1};
			\node[server,name=S2] at (2,1) {2};
			\node[server,name=S3] at (4,0) {3};	
			\draw[thick, ->, red] (-1, -0.2) node [left] {$1$}-- (0, -0.2)-- (2, -1.2) -- (4, -0.2) -- (5, -0.2);
			\draw[thick, ->, blue] (1, 1.2) node [left] {$2'$} -- (2, 1.2) -- (4, 0.2) -- (5, 0.2);
			\draw[thick, ->, blue] (-1, 0.2)  node [left] {$2$} -- (1, 0.2);
		\end{tikzpicture}
	\end{minipage}\hfill
	\begin{minipage}{0.48\textwidth}
		\centering
		\begin{tikzpicture}
			[server/.style={shape=rectangle,draw,minimum height=.8cm,inner xsep=3ex}]
			\node[server,name=S0] at (0,0) {0};
			\node[server,name=S1] at (2,-1) {1};
			\node[server,name=S2] at (2,1) {2};
			\node[server,name=S3] at (4,0) {3};	
			\draw[thick, ->, red] (-1, -0.2) node [left] {$1$}-- (0, -0.2)-- (2, -1.2) -- (4, -0.2) -- (5, -0.2);
			\draw[thick, ->, blue] (-1, 0.2)  node [left] {$2$} -- (0, 0.2)-- (1, 1.2) -- (3, 1.2);
			\draw[thick, ->, blue] (3, 0.2)  node [left] {$2'' $} -- (5, 0.2);
		\end{tikzpicture}
		
	\end{minipage}
	\caption{Two possibilities for cutting flows to transform the diamond network into a tree-reducible network.}
	\label{fig:ff-cut}
\end{figure}
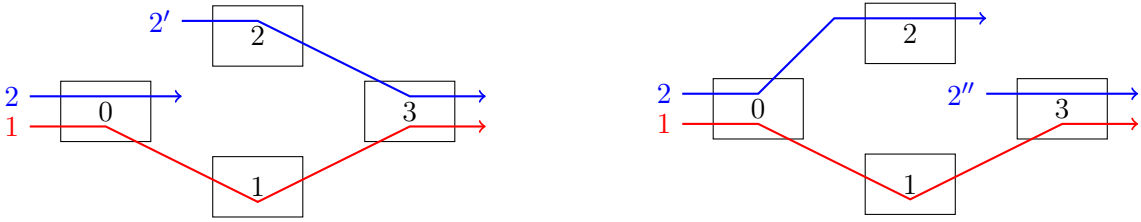

Now, when analysing flow~1, we do not need to consider server~2 (a simple case of tree reduction). The dependency graph of the processes has one non-trivial connected component $\{A_1, A_2, A_{2''}, S_0\}$, and the other components $\{S_1\}$ and $\{S_3\}$ are singletons. 
In the last step, in order to compute the performance bounds for flow~1, we first obtain the bounding generating function of the end-to-end service with Theorem~\ref{th:e2e-gs-dep}: 
\begin{align*} 
	F_{S_{\mathrm{e2e}}}(\theta, z) = & \frac{e^{\theta(\sigma_{A_1}(p_{2''}\theta) + \sigma_{A_2}(p_2\theta) + \sigma_{A_2}(p_{2''} \theta) + \sigma_{S_0}(q_0\theta) + \sigma_{S_0}(p_{2''}\theta) + \sigma_{S_1}(\theta) + \sigma_{S_2}(p_{2''}\theta)+ \sigma_{S_3}(\theta))}}{(1-e^{-p_{2''}\theta(\rho_{S_0}(p_{2''}\theta) - \rho_{A_1}(p_{2''}\theta) - \rho_{A_2}(p_{2''}\theta))})^{\frac{1}{p_{2''}}}(1-e^{-p_{2''}\theta(\rho_{S_2}(p_{2''}\theta) - \rho_{A_2}(p_{2''}\theta))})^{\frac{1}{p_{2''}}}} \\ 
	& \hspace{3mm} \cdot \frac{1}{(1-e^{-\theta(\rho_{S_0}(q_0\theta) - \rho_{A_2}(p_2\theta))}z) (1- e^{-\theta\rho_{S_1}(\theta)}z) (1- e^{-\theta(\rho_{S_3}(\theta) - \rho_{A_2}(p_{2''}\theta))}z)}.
\end{align*}
 Theorem~\ref{th:dep} can now be applied to compute, for instance, the delay violation probability with the relation $1 / p_1 + 1 / p_2 + 1 / p_{2''} + 1 / q_0 = 1.$ 

\subsubsection{Numerical evaluation}

Here, we provide some numerical results for delay bounds for the diamond network from Figure~\ref{fig:toy-ff}.
We compare a sequential analysis, representing the state of the art, with the aforementioned cutting technique to enable the PMOO analysis;  simulation results are provided as well.
The results are displayed in Figure~\ref{fig:diamond-network-results}.

We can observe that both cuts followed by PMOO-AC calculations lead to significantly tighter delay bounds than the standard approach.
To be precise, for a violation probability of $ 10^{-6} $ the cut between servers 2 and 3 (called Cut~(2,3) in Figure~\ref{fig:diamond-network-results}) as well as Cut~(0,2) lead to a delay bound of 14, while the standard bound results in 27.
Moreover, the cutting-based PMOO-AC is able to capture almost the same scaling as the simulations, which is in clear contrast to the standard bound.
It is striking that the gap to the simulation does not increase significantly compared to the independent case, despite the application of Hölder's inequality.
This indicates that the inaccuracy induced by Hölder's inequality cannot only be assessed by the number of its applications, but requires a detailed look at the dependency structure. 

\begin{figure*}[t]
	\centering
	\subfloat[\centering Exponential distribution with~$~\lambda_1 = \lambda_2 = 1.8 $ ]{
		\includegraphics[width=0.3\textwidth]{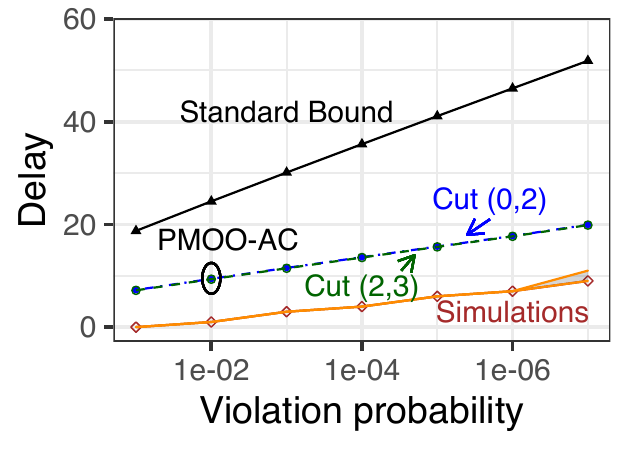}
	}
	\subfloat[\centering Weibull distribution with~$~\lambda_1 = \lambda_2 = 0.9 $]{
		\includegraphics[width=0.3\textwidth]{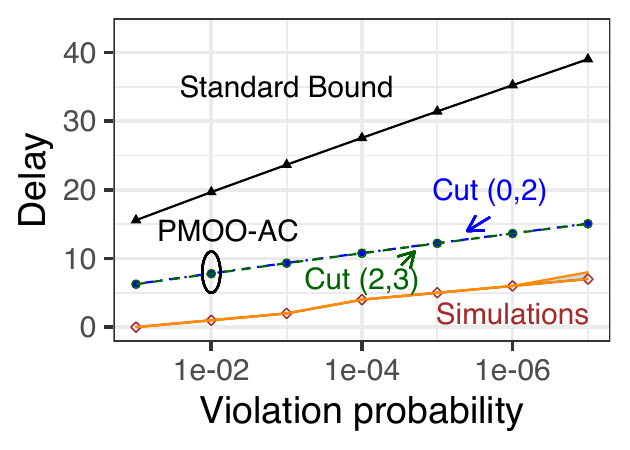}
	}
	\subfloat[\centering MMOO with $ p_{\mathrm{on}, j} = 0.5,$ $p_{\mathrm{off}, j} = 0.5, $ $ P_j = 1.3, $ $ j \in \{1,2\} $]{
		\includegraphics[width=0.3\textwidth]{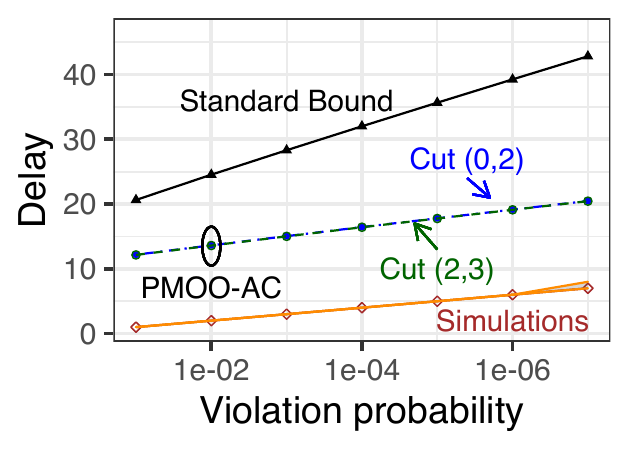}
	}
	\vspace{-2mm}
	\caption{Delay bounds and simulations for the diamond network with server rates $ C_i = 2.0 $ for $ i = 1, \dots, 4 $. \label{fig:diamond-network-results}}
	\vspace{-3mm}
\end{figure*}

\subsection{Discussion}
\label{sec:discussion}
In this section, we have compared state-of-the-art bounds with our new results for a canonical network with rejoining flows;  
here are only two options to cut the network, leading to almost the same bounds. 
While the cutting technique can generally be applied in larger networks, it may no longer be an option to simply try all (combinations of) cuts as these generally grow exponentially in the size of the network. 
A clever search for "good" cuts is conceivable -- similar to advanced network analysis techniques in deterministic network calculus, for instance using machine learning techniques as in \cite{GB19} -- but left for future work. 

	\begin{figure}[b]
	\centering
	\begin{tikzpicture}
		[ server/.style={shape=rectangle,draw,minimum height=.8cm,inner xsep=3ex}]
		\node[server,name=S01] at (0,-1) {$0$};
		\node[server,name=S02] at (0,1) {$0'$};
		\node[server,name=S1] at (2,-1) {$1$};
		\node[server,name=S2] at (2,1) {$2$};
		\node[server,name=S3] at (4,0) {$3$};
		\draw[thick, ->, red] (-1, -0.8) node [left] {$1$} -- (2, -0.8 )-- (4, -0.2) -- (5, -0.2);
		\draw[thick, ->, blue] (-1, -1.2) node [left] {$2'$}-- (1, -1.2); 
		\draw[thick, ->, blue] (-1, 0.8) node [left] {$2$}-- (2, 0.8 )-- (4, 0.2) -- (5, 0.2);
		\draw[thick, ->, red] (-1, 1.2)node [left] {$1'$} -- (1, 1.2);
	\end{tikzpicture}
	\caption{Unfolding of the diamond network of Figure~\ref{fig:toy-ff}.}
	\label{fig:unfold}
\end{figure}
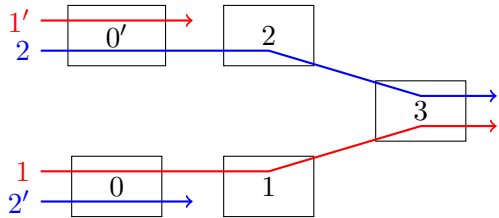

In fact, other strategies to transform a feedforward into a tree-reducible network exist and may lead to more accurate or computationally efficient bounds. For example, a so-called {\em unfolding}, depicted in Figure~\ref{fig:unfold} for the diamond example. 
Server 0 is duplicated in servers 0 and 0', flow 1 (resp. flow 2) is duplicated in flows 1 and 1' (resp. flows 2 and 2'). 
The corresponding arrival and service are the same, and are therefore dependent. 
Under the assumption of independence of flows and servers, we obtain a dependence graph with non-trivial connected components  $G_0 = \{S_0, S_{0'}\}$, $G_1 = \{A_1, A_{1'}\}$, $G_2 = \{A_2, A_{2'}\}$. 
The connected components are singletons.

%


To compute the performance bounds for flow 1, we can apply Theorem~\ref{th:dep} using the following bounding generating function of the end-to-end server:
\begin{align*} 
	F_{S_{\mathrm{e2e}}}(\theta, z) =&  \frac{e^{\theta(\sigma_{A_1}(p_{1'}\theta) + \sigma_{A_2}(p_2\theta) + \sigma_{A_2}(p_{2'}\theta) + \sigma_{S_0}(q_0\theta) + \sigma_{S_0}(q_{0'}\theta)  + \sigma_{S_1}(\theta) + \sigma_{S_2}(\theta)+ \sigma_{S_3}(\theta))}}{(1-e^{-\theta(\rho_{S_0}(q_{0'}\theta) - \rho_{A_1}(p_{1'}\theta) - \rho_{A_2}(p_2\theta))})(1-e^{-\theta(\rho_{S_2}(\theta) - \rho_{A_2}(p_2\theta))})} \\ 
	& \hspace{3mm} \cdot \frac{1}{(1 -  e^{-\theta(\rho_{S_0}(q_{0}\theta) - \rho_{A_2}(p_{2'}\theta))}z) (1-  e^{-\theta \rho_{S_1}(\theta)}z)(1 -  e^{-\theta(\rho_{S_3}(\theta) - \rho_{A_2}(p_2\theta))}z)},
\end{align*}
%
with the relations  $1/q_{0} + 1/q_{0'} = 1$, $1/p_{1} + 1/p_{1'} = 1$ and $1/p_{2} + 1/p_{2'} = 1$.

Even in the small diamond network, it is not simple to carry out an analytical comparison between cutting and unfolding. 
On the one hand,  $(1- e^{-pu})^{\frac{1}{p}} \geq 1-e^{-u}$ for $u>0$ and $p>1$, so cutting has an advantage. 
On the other hand,  unfolding provides more freedom in the H\"older parameters: if one sets $p^C_{2''} = q^U_{0'} = p^U_{1'} = p^U_2$, with the superscripts differentiating between parameters from the unfolding ($U$) and cutting ($C$) methods, one obtains $q^U_{0} < q^C_0$, $p^U_1 <  q^C_1$, and $p^U_{2'} < p^C_2$, and unfolding seems advantageous, at least from the point of view of the stability condition.

In our numerical case study for the diamond network, cutting actually outperformed unfolding. Yet, we speculate that the unfolding method could lead to better bounds in larger networks. To that end, we remark that in deterministic network calculus unfolding in some contexts leads to better performance bounds than cutting \cite{Bou20}. 

In conclusion, it is clear that, while we have "opened the door" for an accurate and efficient SNC analysis of general feedforward networks, there is now a "large room" to explore for future work.

	\section{Conclusion}
\label{sec:conclusion}

We have presented a new network analysis method that unleashes the power of the pay multiplexing only once (PMOO) principle in the stochastic network calculus. 
Based on this method, we applied techniques from analytic combinatorics to keep bounds accurate even in rather complex scenarios.  
Equipped with this, it is now possible to calculate rigorous probabilistic performance bounds for tree-reducible networks without incurring any method-pertinent stochastic dependencies.
In numerical evaluations, we observed that we are largely successful in not widening the known simulation-calculation gap further, and, at least, closely capture the scaling of the performance bounds.
We have also made a promising step towards a stochastic network calculus analysis of general feed-forward networks laying the foundation to reuse the PMOO as much as possible. 

While our method can benefit from improvements based on a preliminary network transformation (\textit{e.g.} flow prolongation \cite{NS20}) at no cost, we believe it can also exploit other recent techniques such as the $h$-mitigators, which were successfully applied to sink-tree networks in~\cite{NSS19}. 
As discussed at the end of Section~\ref{sec:dependent}, for future work it is very promising to invest more effort in good strategies for the transformation of large general feedforward networks into tree-reducible ones.
More disruptively, the PMOO method could also be a first step to enable martingale techniques (as in \cite{CPS14, PC14}) in the end-to-end analysis, in order to completely close the simulation-calculation gap. 

\begin{acks}
	We thank the anonymous reviewers and our shepherd Mark S. Squillante for their insightful feedback and guidance.
	This work was partially supported by Huawei Technologies Co., Ltd.
\end{acks}

\bibliographystyle{plain}
\bibliography{biblio.bib}
	
\appendix

\renewcommand{\thesection}{\Alph{section}}

	\section{Additional Proofs of Section~\ref{sec:framework}}
	\subsection{Proof of Equation~\eqref{eq:elie}}
	\label{sec:elie}
		\begin{align*}
			G(z) & = \sum_{n = 0}^\infty (\sum_{m = n}^\infty f_m) z^n 
			= \sum_{n=0}^\infty (\sum_{m =0}^\infty f_m)z^n - \sum_{n=0}^\infty (\sum_{m=0}^{n-1} f_m)z^n \\
			&= F(1) \sum_{n = 0}^\infty z^n - \sum_{n=1}^\infty (\sum_{m=0}^{n-1} f_m) z^n 
			= \frac{F(1)}{1-z} - z  \sum_{n=0}^\infty (\sum_{m=0}^{n} f_m) z^n \\ 
			&= \frac{F(1)}{1-z} - z  \sum_{n=0}^\infty (\sum_{m=0}^n f_mz^m z^{n-m}) 
			= \frac{F(1)}{1-z} - z  \sum_{m = 0}^\infty (f_m z^m \sum_{n = m}^\infty z^{n-m}) \\
			&= \frac{F(1)}{1-z} - z  \frac{F(z)}{1-z}
			= \frac{F(1)- zF(z)}{1-z}.
		\end{align*}

\subsection{Proof of Lemma~\ref{lem:delay}}
\label{sec:delay}
As $F_A(\theta, z)$ is a geometric series, we have $[z^{n+m}]F_A(\theta, z) = e^{\theta\rho_A(\theta) m} [z^{n}]F_A(\theta, z)$. 
 Let us denote $d(t)$ the delay at time $t$. We have for all $T>0$,
 \begin{align}\p(d(t) \geq T) & = \p(A(0, t) > D(0, t + T - 1))& \notag \\ 
 	& \leq \p(\exists s \leq t,~A(s, t) > S(s, t + T - 1)) &\notag\\
 	& \leq \sum_{0\leq s\leq t}\p(A(s, t) > S(s, t+T - 1))&\notag\\
 	&\leq \sum_{0\leq s\leq t-1} \E[e^{\theta (A(s, t) - S(s, t+T - 1))}]& \text{($A(t, t) = 0$ and $S(t, t+T-1) \geq 0$)}\notag\\
 	&= \sum_{0\leq s\leq t-1} \E[e^{\theta A(s, t)}] \E[e^{- \theta S(s, t+T - 1)}]& \text{(independence of $A$ and $S$)}\notag\\
 	&\leq \sum_{0\leq s\leq t-1} [z^{t-s}]F_A(\theta, z) \cdot [z^{t-s+T - 1}]F_S(\theta, z)&\notag\\
 	&\leq \sum_{u >  0}  [z^{u}]F_A(\theta, z) \cdot [z^{u+T - 1}]F_S(\theta, z)& (u\leftarrow t-s)\notag\\
 	&= e^{-\theta\rho_A(\theta) (T-1)}\sum_{u> 0} [z^{u+T - 1}]F_A(\theta, z) \cdot [z^{u+T - 1}]F_S(\theta, z)&\notag\\
 	&= e^{-\theta\rho_A(\theta) (T-1)}\sum_{u \geq  T} [z^{u}]F_A(\theta, z) \cdot [z^{u}]F_S(\theta, z).& \label{eq:delay}
 \end{align}
%
We recognize  the sum of the last terms of a Hadamard product in the right-hand term. Let $h_u = [z^{u}]F_A(\theta, z) \cdot [z^{u}]F_S(\theta, z)$ and $h'_u = \sum_{u\geq T} h_u$. Let $H(z)$ and $H'(z)$ be their corresponding generating functions.  In particular, we have $H(z) = e^{\theta \sigma_A(\theta)}F_S(\theta, e^{\theta\rho_A(\theta)}z)$.
Now, $$\p(d(t) \geq T) \leq e^{-\theta\rho_A(\theta) (T-1)} [z^T] H'(z),$$ and we recognize again a Hadamard product of a geometric series and $H'$, and use Equation~\eqref{eq:elie} to obtain  
\begin{align*}
	F_d(z) & = e^{\theta\rho_A(\theta)} H'(e^{-\theta\rho_A(\theta)}z) 
	= e^{\theta\rho_A(\theta)}\frac{H(1) - e^{-\theta\rho_A(\theta)}zH(e^{-\theta\rho_A(\theta)}z)}{1-e^{-\theta\rho_A(\theta)}z}\\
	&= e^{\theta (\sigma_A(\theta) + \rho_A(\theta))}\frac{ F_S(\theta, e^{\theta\rho_A(\theta)})- e^{-\theta\rho_A(\theta)}zF_S(\theta, z)}{1-e^{-\theta\rho_A(\theta)}z}.
\end{align*}

\section{Proofs and Computations of Section~\ref{sec:pmoo}}
\label{sec:exact_comp} 
\subsection{Proof of Corollary~\ref{cor:simple-pmoo}}
\label{app:simple-pmoo}
%


It suffices to prove that for all $t \in\N$, $[z^t]F_{S_{\mathrm{e2e}}}(\theta, z)\leq [z^t]G(\theta, z)$. For all $t\in \N$, 
\begin{align*}
	[z^t]F_{S_{\mathrm{e2e}}}(\theta, z) = & e^{\theta\sigma_{S_{\mathrm{e2e}}}(\theta)} \sum_{u_1 + \ldots + u_n = t} \prod_{j=1}^n e^{-\theta \rho'_j(\theta)u_j}\\
	= & 	 e^{\theta\sigma_{S_{\mathrm{e2e}}}(\theta)} e^{-\theta\rho'_1(\theta)t}  \sum_{u_1 + \ldots + u_n = t} \prod_{j=1}^n e^{-\theta (\rho'_j(\theta) - \rho'_1(\theta))u_j}\\
	= & e^{\theta\sigma_{S_{\mathrm{e2e}}}(\theta)} e^{-\theta\rho'_1(\theta)t} \sum_{s_1 + s_2 = t}   \left(\sum_{u_1 + \ldots + u_k = s_1}1\right)  \\ & \hspace{2cm} \cdot\left(\sum_{u_{k+1} + \ldots + u_n = s_2} \prod_{j=k+1}^n e^{-\theta (\rho'_j(\theta) - \rho'_1(\theta))u_j}\right)\\
	= & e^{\theta\sigma_{S_{\mathrm{e2e}}}(\theta)} e^{-\theta\rho'_1(\theta)t} \sum_{s_1 + s_2 = t} \binom{s_1 + k -1}{k-1} \\ & \hspace{2cm} \cdot \left(\sum_{u_{k+1} + \ldots + u_n = s_2} \prod_{j=k+1}^n e^{-\theta (\rho'_j(\theta) - \rho'_1(\theta))u_j}\right)\\
	\leq & e^{\theta\sigma_{S_{\mathrm{e2e}}}(\theta)} e^{-\theta\rho'_1(\theta)t} \binom{t + k -1}{k-1}  \left(\sum_{u_{k+1} + \ldots + u_n \leq t} \prod_{j=k+1}^n e^{-\theta (\rho'_j(\theta) - \rho'_1(\theta))u_j}\right)\\
	\leq & e^{\theta\sigma_{S_{\mathrm{e2e}}}(\theta)} e^{-\theta\rho'_1(\theta)t} \binom{t + k -1}{k-1}   \prod_{j=k+1}^n \frac{1}{1-e^{-\theta (\rho'_j(\theta) - \rho'_1(\theta))}} = [z^t]G(\theta, z). 
\end{align*}

\subsection{Proof of Equation~\eqref{eq:sing1}: $F_{S_{\mathrm{e2e}}}(\theta, z)$ has singularities with multiplicities 1 only.}
\label{app:sing1}
	
We first use partial fraction decomposition~\cite{ChangFC73}:
$$\prod_{j=1}^n\frac{1}{1-r_jz} =\sum_{j = 1}^n \prod_{k\neq j}\Big( \frac{1}{1-r_j^{-1}r_k}\Big) \frac{1}{1-r_jz}.$$
and
\begin{multline*}
	\frac{r_0^{-1}}{1-r_0z} \prod_{j=1}^n\frac{1}{1-r_0^{-1}r_j} -
 \frac{z}{1-r_0z}  \prod_{j=1}^n\frac{1}{1-r_jz} \\ = \frac{r_0^{-1}}{1-r_0z} \prod_{j=1}^n\frac{1}{1-r_0^{-1}r_j} -  z\Big[\sum_{j = 0}^n \Big(\prod_{k \geq 0, k\neq j} \frac{1}{1-r_j^{-1}r_k}\Big) \frac{1}{1-r_jz}\Big]
	\\  = r_0^{-1}\prod_{j=1}^n\frac{1}{1-r_0^{-1}r_j} - \sum_{j=1}^n \frac{z}{1-r_0r_j^{-1}}\Big( \prod_{k\geq 1, k\neq j} \frac{1}{1-r_j^{-1}r_k}\Big) \frac{1}{1-r_jz}.
\end{multline*}

In the bounding generating function of the delay, by identifying $r_0$ as $e^{-\theta \rho_{A_1}(\theta)}$ and for all $j>0$, $r_j$ as $e^{-\theta \rho'_j(\theta)}$, we obtain (for ease of presentation, we drop the dependence of the $\sigma$'s and $\rho$'s in $\theta$)
\begin{align*}
	F_d(\theta, z) & = e^{\theta \sigma_{A_1}}\frac{ e^{\theta\rho_{A_1}}F_{S_{\mathrm{e2e}}}(\theta, e^{\theta\rho_{A_1}})- zF_{S_{\mathrm{e2e}}}(\theta, z)}{1-e^{-\theta \rho_{A_1}}z}\\
	& = \frac{e^{\theta(\sigma_{A_1}  + \sigma_{S_{\mathrm{e2e}}} + \rho_{A_1})}}{\prod_{j=1}^n\big( 1-e^{\theta(\rho_{A_1}  - \rho'_j)}\big)} \frac{1}{1-e^{-\theta\rho_{A_1}}z} - \frac{e^{\theta(\sigma_{A_1}  + \sigma_{S_{\mathrm{e2e}}})}z}{1-e^{-\theta\rho_{A_1}}z}\prod_{j=1}^n \frac{1}{1-e^{-\theta\rho'_j}z}\\
	&= \frac{e^{\theta(\sigma_{A_1}  + \sigma_{S_{\mathrm{e2e}}}+ \rho_{A_1})}}{\prod_{j=1}^n \big(1-e^{\theta(\rho_{A_1}  - \rho'_j)}\big)} + \sum_{j=1}^n \frac{e^{\theta(\sigma_{A_1}  + \sigma_{S_{\mathrm{e2e}}})}z}{e^{\theta(\rho'_j- \rho_{A_1})}-1} \Big(\prod_{k\neq j} \frac{1}{1-e^{\theta(\rho'_j - \rho'_k)}} \Big)\frac{1}{1-e^{-\theta\rho'_j}z}. 
\end{align*}
This is of the form $f(z) = a + \sum_{j=1}^n \frac{b_jz}{1-r_j z}$, and the coefficients of $f$ are $[z^0]f(z) = f(0) = a$ and for all $T>0$, 
\begin{equation}
	\label{eq:almost-geometric}
	[z^T]f(z) = \sum_{j=1}^n [z^T]\frac{b_jz}{1-r_jz} = \sum_{j=1}^n  [z^{T-1}]\frac{b_j}{1-r_jz} =  \sum_{j=1}^n b_j r_j^{T-1}.
\end{equation}
As a consequence, for all $T>0$,

\begin{align*}
	[z^T]F_d(\theta, z) & = \sum_{j=1}^n \frac{e^{\theta(\sigma_{A_1}  + \sigma_{S_{\mathrm{e2e}}})}}{e^{\theta(\rho'_j- \rho_{A_1})}-1} \Big(\prod_{k\neq j} \frac{1}{1-e^{\theta(\rho'_j - \rho'_k)}}\Big) e^{-\theta\rho'_j(T-1)} \\
	& =  \sum_{j=1}^n \frac{e^{\theta(\sigma_{A_1}  + \sigma_{S_{\mathrm{e2e}}} + \rho_{A_1})}}{1-e^{\theta(\rho_{A_1} - \rho'_j)}} \Big(\prod_{k\neq j} \frac{1}{1-e^{\theta(\rho'_j - \rho'_k)}}\Big) e^{-\theta\rho'_jT}. 
\end{align*}

\subsection{Proof of Equation~\eqref{eq:uniform}: $F_{S_{\mathrm{e2e}}}(\theta,z)$ has exactly one singularity}
\label{app:uniform}
In this subsection, we assume that all $\rho'_j(\theta)$, $j\in\{1, \ldots, n\}$ are equal. 

We also use partial fractional decomposition: 
$$\frac{1}{(1-r_0z)(1-r_1z)^n} = \frac{1}{(1-r_0z)(1-r_0^{-1}r_1)^n} - \frac{r_0^{-1}r_1}{(1-r_0^{-1}r_1)^{n+1}} \left[\sum_{i=1}^{n}\left( \frac{1-r_0^{-1}r_1}{1-r_1z}\right)^{i}\right],$$
and 
\begin{multline*}
	\frac{r_0^{-1}}{(1-r_0z)(1-r_0^{-1}r_1)^n} - \frac{z}{(1-r_0z)(1-r_1z)^n} = \\ \frac{r_0^{-1}}{(1-r_0^{-1}r_1)^n} + \frac{r_0^{-1}r_1z}{(1-r_0^{-1}r_1)^{n+1}} \left[\sum_{i=1}^{n}\left( \frac{1-r_0^{-1}r_1}{1-r_1z}\right)^{i}\right],
\end{multline*}

The bounding generating function of the delay can be obtained by replacing $r_0$ by $e^{-\theta\rho_{A_1}}$ and $r_1$ by $e^{-\theta\rho'_{1}}$: 
\begin{multline*}
	F_d(\theta, z)  = e^{\theta(\sigma_{A_1} + \sigma_{S_{\mathrm{e2e}}})} \left(\frac{e^{\theta\rho_{A_1}}}{(1-e^{-\theta(\rho'_1 - \rho_{A_1})})^n} \right. \\ \left.+ e^{\theta(\rho_{A_1} - \rho'_{1})}z \left[\sum_{i=1}^{n}\frac{1}{(1-e^{-\theta(\rho'_1 - \rho_{A_1})})^{n-i+1}}  \frac{1}{(1-e^{-\theta\rho'_1}z)^{i}}\right]\right).
\end{multline*}

If $f(z) = \frac{1}{(1-rz)^m}$, then its coefficients are $[z^T]f(z) = \binom{T+m-1}{T} r^T$. Then, following the same computations as in~\eqref{eq:almost-geometric}, one can deduce that for all $T>0$,
\begin{align*}
	[z^T]F_d(\theta, z) &=e^{\theta(\sigma_{A_1} + \sigma_{S_{\mathrm{e2e}}} + \rho_{A_1} - \rho'_{1})}	\left[\sum_{i=1}^{n}\frac{1}{(1-e^{-\theta(\rho'_1 - \rho_{A_1})})^{n-i+1}}  \binom{T + i-2}{T-1}e^{-\theta\rho'_{1}(T-1)} \right]\\
	& = e^{\theta(\sigma_{A_1} + \sigma_{S_{\mathrm{e2e}}} + \rho_{A_1})}	\left[\sum_{i=1}^{n}\frac{1}{(1-e^{-\theta(\rho'_1 - \rho_{A_1})})^{n-i+1}}  \binom{T + i-2}{T-1}e^{-\theta\rho'_{1}T} \right].
\end{align*}

\section{Proof of Theorem~\ref{th:dep}}
\label{app:dependent}
\subsection{Bounding generating function of the departure process}
We have
\begin{align*}
	A_1\deconv S_{\mathrm{e2e}} (s, t) & = \sup_{u\leq s} A_1(u, t) - S_{\mathrm{e2e}}(u, s)\\
	& = \sup_{u\leq s}\sup_{\substack{t_j\leq t_{j^{\bu}}\\ t_{\pi_1(1)} = u \\ t_{\pi_1(\ell_1)^{\bu}} = s}} A_1(u, t) - \sum_{j=1}^n S_j(t_j, t_{j^\bu}) + \sum_{i=2}^m A_i(t_{\pi_i(1)}, t_{\pi_i(\ell_i)^\bu}).
\end{align*}

When computing the MGF $\E[e^{\theta (A_1\deconv S_{\mathrm{e2e}} (s, t))}]$, one can take advantage of the partial independence of $A_1$ with the other processes. 

For all $(p_i), (q_j)$ such that for all $k\in\{1, \ldots, K\}$, $\sum_{A_i\in G_k} \frac{1}{p_i} + \sum_{S_j\in G_k} \frac{1}{q_j}=1$, we have, using Lemma~\ref{lem:partialH} in the second line, 

\begin{align*}
	\E[e^{\theta (A_1\deconv S_{\mathrm{e2e}} (s, t))}]  
	& \leq \sum_{u\leq s}\sum _{\substack{t_j\leq t_{j^{\bu}}\\ t_{\pi_1(1)} = u\\ t_{\pi_1(\ell_1)^{\bu}} = s}} \E[e^{\theta( A_1(u, t) - \sum_{j=1}^n S_j(t_j, t_{j^\bu}) + \sum_{i=2}^m A_i(t_{\pi_i(1)}, t_{\pi_i(\ell_i)^\bu}))}]\\
	& \leq \sum_{u\leq s}\sum _{\substack{t_j\leq t_{j^{\bu}}\\ t_{\pi_1(1)} = u\\ t_{\pi_1(\ell_1)^{\bu}} = s}} e^{\theta(\sigma_{A_1}(p_1\theta) + \rho_{A_1}(p_1\theta) (t-u))}\left(\prod_{j = 1}^n e^{-\theta(\sigma_{S_j}(q_j\theta) - \rho_{S_j}(q_j\theta)(t_{j^{\bu}} - t_j))}\right.  \\ & \hspace{5cm} \left.\prod_{i=2}^m e^{\theta(\sigma_{A_i}(p_i\theta) + \rho_{A_i}(p_i\theta)(t_{\pi_i(\ell_i)^\bu}- t_{\pi_i(1)}))}\right)\\
	& = \sum_{u\leq s} e^{\theta(\sigma_{A_1}(p_1\theta) + \rho_{A_1}(p_1\theta) (t-u))} \cdot \sum _{\substack{t_j\leq t_{j^{\bu}}\\ t_{\pi_1(1)} = u\\ t_{\pi_1(\ell_1)^{\bu}} = s}}  \left( \prod_{j=1}^n e^{-\theta(\sigma_{S_j}(q_j\theta) - \rho_{S_j}(q_j\theta)(t_{j^{\bu}} - t_j))}\right. \\*
	 &
	\hspace{5cm} \left.\prod_{i=2}^m e^{\theta(\sigma_{A_i}(p_i\theta) + \rho_{A_i}(p_i\theta)(t_{\pi_i(\ell_i)^\bu}- t_{\pi_i(1)}))}\right)\\
	& = \sum_{u\leq s}[z^{t-u}]F_{A_1}(\theta, p_1, z) [z^{s-u}]F_{S_{\mathrm{e2e}}}(\theta, (p_i), (q_j), z) .
\end{align*}

Similar calculations to those in Lemma~\ref{lem:output} lead to  
\begin{align*}
	F_{D_{\mathrm{e2e}}}(\theta, (p_i),(q_j), z) & = \frac{e^{\theta \sigma_{A_1}(p_1\theta)} F_{S_{\mathrm{e2e}}}(\theta, (p_i)_{i\neq 1}, (q_j), e^{\theta \rho_{A_1}(p_1\theta)})}{1-e^{\theta \rho_{A_1}(p_1\theta)}z}.
\end{align*}


\subsection{Bounding generating function of the delay}
From the proof of Lemma~\ref{lem:delay} and then Lemma~\ref{lem:partialH}, 
\begin{align*}
	\p(d(t) \geq T)  &\leq \sum_{0\leq s\leq t - 1} \E[e^{\theta (A_1(s, t) - S_{\mathrm{e2e}}(s, t+T-1))}]\\
	&\leq \sum_{0\leq s\leq t-1} \sum _{\substack{t_j\leq t_{j^{\bu}}\\ t_{\pi_1(1)} = s\\ t_{\pi_1(\ell_1)^{\bu}} = t+T - 1 }} \E[e^{\theta( A_1(s, t) - \sum_{j=1}^n S_j(t_j, t_{j^\bu}) + \sum_{i=2}^m A_i(t_{\pi_i(1)}, t_{\pi_i(\ell_i)^\bu}))}]\\
	& \leq \sum_{s\leq t-1}\sum _{\substack{t_j\leq t_{j^{\bu}}\\ t_{\pi_1(1)} = s\\ t_{\pi_1(\ell_1)^{\bu}} = t+T-1}}  e^{\theta(\sigma_{A_1}(p_1\theta) + \rho_{A_1}(p_1\theta) (t-s))}  \left(\prod_{j=1}^n e^{-\theta(\sigma_{S_j}(q_j\theta) - \rho_{S_j}(q_j\theta)(t_{j^{\bu}} - t_j))}\right. \\ 
	& \hspace{5cm} \left.\prod_{i=2}^m e^{\theta(\sigma_{A_i}(p_i\theta) + \rho_{A_i}(p_i\theta)(t_{\pi_i(\ell_i)^\bu}- t_{\pi_i(1)}))}\right)\\
	& \leq \sum_{s\leq t-1} [z^{t-s}]F_{A_1}(\theta, p_1, z) \cdot [z^{t+T-s-1}]F_{S_{\mathrm{e2e}}}(\theta,(p_i)_, (q_j), z).
\end{align*}
Following the same lines as in Lemma~\ref{lem:delay} leads to 
\begin{multline*}
	F_{d}(\theta, (p_i), (q_j), z) = \frac{{e^{\theta (\sigma_{A_1}(p_1\theta) + \rho_{A_1}(p_1\theta))}}( F_{S_{\mathrm{e2e}}}(\theta, (p_i)_{i\neq 1}, (q_j), e^{\theta \rho_{A_1}(p_1\theta)})}{1-e^{-\theta \rho_{A_1}(p_1\theta)}z} - \\ \frac{ e^{\theta \sigma_{A_1}(p_1\theta)}z F_{S_{\mathrm{e2e}}}(\theta, (p_i)_{i\neq 1}, (q_j), z))}{1-e^{-\theta \rho_{A_1}(p_1\theta)}z}.
\end{multline*}

%
%
%
%
%
%
%
%
%
%

\end{document}